\documentclass{amsart}
\usepackage{amsmath}
\usepackage{amssymb}
\usepackage{latexsym}
\usepackage{amsfonts}
\usepackage[dvips]{graphicx}
\usepackage{color}
\usepackage{graphicx,epstopdf}

\newtheorem{thm}{Theorem}

\newtheorem{conj}{Conjecture}

\usepackage{float}
\usepackage[section]{placeins}

\renewcommand{\a}{\alpha}

\renewcommand{\tilde}{\widetilde}

\usepackage{color}
\usepackage{graphicx}
\usepackage{epstopdf}

\begin{document}

\title[On deformation and classification of $\vee$-systems]{On deformation and classification of $\vee$-systems}

\author{V. Schreiber}\address{Department of Mathematical Sciences,
Loughborough University, Loughborough LE11 3TU, UK}
\email{V.Schreiber@lboro.ac.uk}
\author{A.P. Veselov}
\address{Department of Mathematical Sciences,
Loughborough University, Loughborough LE11 3TU, UK  and Moscow State University, Moscow 119899, Russia}
\email{A.P.Veselov@lboro.ac.uk}

\maketitle

\begin{abstract}
The $\vee$-systems are special finite sets of covectors which appeared in the theory of the generalized
Witten-Dijkgraaf-Verlinde-Verlinde (WDVV) equations. Several families of $\vee$-systems are known, but their classification is an open problem.
We derive the relations describing the infinitesimal deformations of $\vee$-systems 
and use them to study the classification problem for $\vee$-systems in dimension three. 
We discuss also possible matroidal structures of $\vee$-systems in relation with projective geometry and give the catalogue of all known irreducible rank three $\vee$-systems.
\end{abstract}


\section{Introduction}

The $\vee$-systems are special finite sets of covectors introduced in \cite{V1,V2}. The motivation came from the study of certain special solutions of the generalized Witten-Dijkgraaf-Verlinde-Verlinde (WDVV) equations, playing an important role in 2D topological  field theory and $N=2$ SUSY Yang-Mills theory \cite{D1, MMM}. 

Let $V$ be a real vector space and  $\mathcal{A}\subset V^*$ be a
finite set of vectors in the dual space $V^*$ (covectors) spanning $V^*$. To such a set one can
associate the following {\it canonical form} $G_{\mathcal A}$ on
$V$:
$$G_{\mathcal A}(x,y)=\sum_{\a\in\mathcal{A}}\a(x)\a(y), $$
where $x,y\in V$, which
establishes the isomorphism
$$
\varphi_{\mathcal A}: V \rightarrow V^*.
$$
The inverse $\varphi_{\mathcal A}^{-1}(\alpha)$ we denote as
$\a^\vee$.
The system $\mathcal{A}$ is called $\vee$-{\it system} if the following relations 
\begin{equation}
\label{vee}
\sum\limits_{\beta \in \Pi \cap \mathcal {A}}
\beta(\alpha^\vee)\beta^\vee=\nu \alpha^{\vee}
\end{equation}
(called $\vee$-{\it conditions}) are satisfied for any $\alpha \in \mathcal{A}$ and any two-dimensional plane $\Pi \subset V^*$ containing $\alpha$ and some $\nu$, which may depend on $\Pi$ and $\alpha.$ If $\Pi$ contains more than two covectors, then $\nu$ does not depend on $\alpha \in \Pi$ and the corresponding two forms $G_{\mathcal A}$ and 
$$G^{\Pi}_{\mathcal A}(x,y):=\sum_{\a\in\Pi\cap\mathcal{A}}\a(x)\a(y)$$ are proportional on the plane $\Pi^{\vee} \subset V$ (see \cite{V1,V2}).
If $\Pi$ contains only two covectors from $\mathcal A$, say $\alpha$ and $\beta,$ then we must have
$$
G_{\mathcal A}(\alpha^{\vee}, \beta ^{\vee})=0.
$$
The $\vee$-conditions are equivalent to the flatness of the corresponding Knizhnik--Zamolodchikov-type $\vee$-connection 
\[
\nabla_a=\partial_a+\kappa\sum_{\alpha\in {\mathcal A}}\frac{\langle\alpha,a\rangle}{\langle\alpha,x\rangle}\alpha^\vee\otimes \alpha.
\]

The examples of $\vee$-systems include
all two-dimensional systems, Coxeter configurations and so-called deformed root systems \cite{MG, SV,V1}, 
but the full classification is an open problem. The main results in this direction can be found in \cite{CV,FV,FV2,F,LST}. In particular, in \cite{FV2} it was shown that the class of $\vee$-systems is closed under the operation of restriction, which gives a powerful tool to construct new  examples of $\vee$-systems.

The most comprehensive list of known $\vee$-systems together with their geometric properties can be found in \cite{FV, FV2}.
The main purpose of this paper is to present some arguments in favour of the completeness of this list in dimension three by studying the infinitesimal deformations of $\vee$-systems. 

We start with a brief review of the general notions from the theory of matroids  \cite{Oxley}, which provides a natural framework for the problem of classification of $\vee$-systems. A matroidal approach in this context was also used by Lechtenfeld et al in \cite{LST}.

Then we study the infinitesimal deformations of the $\vee$-systems of given matroidal type and derive the corresponding linearised $\vee$-conditions.
This allows us to show that the isolated 3D $\vee$-systems listed in \cite{FV} are indeed isolated.

The main question which still remains open is what are possible matroidal structures of $\vee$-systems. We discuss this in the context of the projective geometry using the analysis of the known $\vee$-systems.

In the last section we study the property of the corresponding $\nu$-function on the flats of matroid and state the uniqueness conjecture, saying that the matroid and function $\nu$ on its flats uniquely determine the corresponding $\vee$-system.

In the Appendix we give the catalogue of all known $\vee$-systems in dimension three together with the corresponding matroids and $\nu$-functions.

\section{Vector Configurations and Matroids}

The combinatorial structure of the vector configurations can be described using the notion of matroid.
The theory of {\it matroids} was introduced by Whitney in 1935, who was looking for an abstract notion generalising the linear dependence in the vector space.

We review some standard notions from this theory following mainly Oxley \cite{Oxley}.

 A {\it matroid} M is a pair ($X$, $\mathcal{I}$), where $X$ is a finite set
and $\mathcal{I}$ is a collection of subsets $\mathit{S}$ of X (called the {\it independent sets} of $M$) such that:
\begin{itemize}
\item $\mathcal{I}$ is non-empty
\item For any $\mathit{S}$$\in\mathit{\mathcal{I}}$, any $\mathit{S}{}^{'}\subset\mathit{S}$
one has $\mathit{S}^{'}\in\mathcal{I}$.
\item If $A, B\in\mathcal{I}$, $\mid A\mid=\mid B\mid+1$ then $\exists x\in A\setminus B$
such that $B\cup\{x\}\in\mathcal{I}$. 
\end{itemize}

The {\it rank} of the matroid $M$ is defined as $r\left(M\right)=max_{I\in\mathcal{I}}\{ \mid I\mid\}.$ 
More generally, the rank of the subset $S \subset X$ is defined as $r\left(S\right)=max_{I\in\mathcal{I}}\{ \mid I\mid: I \subseteq S\}.$ 

A {\it direct sum}  of matroids $M_1=(X_1, \mathcal{I}_1)$ and $M_2=(X_2, \mathcal{I}_2)$ is defined
as
$$M_1\oplus M_2=(X_1\cup X_2, \{I_1\cup I_2: I_1 \in \mathcal{I}_1, I_2 \in \mathcal{I}_2\}).$$
A matroid is called {\it connected} if it can not be represented as a direct sum.

The most important class of matroids for us consists of vector matroids. Let $A$ be a real $r \times n$ matrix, $X=\{1,2,\dots, n\}$ be the set of column labels of $A$, and $\mathcal{I}$ be
the collection of subsets $S$ of $X,$ for which the columns labelled by $S$ are linearly independent over $\mathbb R.$ Then ($X$,$\mathcal{I}$) is a matroid, which is called {\it rank $r$ vector matroid} 
and denoted by $M[A].$ 

The following operations on matrix $A$ do not affect the corresponding
vector matroid $M[A]:$ 

\begin{enumerate}
 \item Elementary operations with the rows,
 \item Multiplication of a column by a non-zero number.
\end{enumerate}

Two matrices $A$ and $A^{'}$ representing the same matroid $M$
are said to be {\it projectively equivalent representations} of $M$
if $A^{'}$ can be obtained from $A$ by a sequence of these operations. 
Equivalently, one can say that $A'=CAD$, where $C$ is an invertible $r\times r$ matrix, 
and $D$ is a diagonal $n \times n$ matrix with non-zero diagonal entries.

Alternatively, one can define the {\it linear dependence matroid} on the set $X$ as a family  $\mathcal{I_{C}}$ of minimal dependent subsets $\mathit{C}$ of $X$ (called {\it circuits}) through the following axioms:\
\begin{itemize}
\item The empty set is not a circuit.
\item No curcuit is contained in another circuit.
\item If $C_1,C_2\in\mathcal{I_{C}}$ are two circuits sharing an element $e\in{X}$, then $(C_{1}\cup C_{2})\setminus{e}$ is a circuit or contains a circuit. 
\end{itemize}

 The rank of a circuit is defined as the dimension of the vector space spanned by its vectors. Circuits spanning the same $d$-dimensional subspace can be united in so-called $d$-flats. A set $F\subseteq X$ is a {\it flat} of the matroid $M$
if $$r(F\cup\{x\})=r(F)+1$$ for all $x\in X\setminus F,$ where $r(F)$ is the rank of the flat $F.$ 
The matroid can be labelled by listing all $d$-flats.

As an example consider the positive roots of the $B_{3}$-type system.
The corresponding matrix (with the first row giving the labelling) is 

$$A =\left[\begin{array}{ccccccccc} 1 & 2 & 3 & 4 & 5 & 6 & 7 & 8 & 9\\ 1 & 1 & 0 & 0 & 1 & 1 & 1 & 0 & 0\\ 1 & -1 & 1 & 1 & 0 & 0 & 0 & 1 & 0\\ 0 & 0 & -1 & 1 & 1 & -1 & 0 & 0 & 1 \end{array}\right].$$
Here matroid $M$ is defined on the set $X=\{1,2,3,4,5,6,7,8,9\},$ 
with 2-flats
$$\left\{ (4,1,6),(6,2,3),(4,5,2),(1,5,3)\right\}$$
and
$$\left\{ (3,4,8,9),(1,2,7,8),(5,6,7,9)\right\}$$
with three and four elements respectively.
Together with the 3-flat $X$ this gives the complete list of flats.

Graphically on the projective plane we have

\begin{figure} [H]
\begin{center}
\includegraphics[scale=0.3]{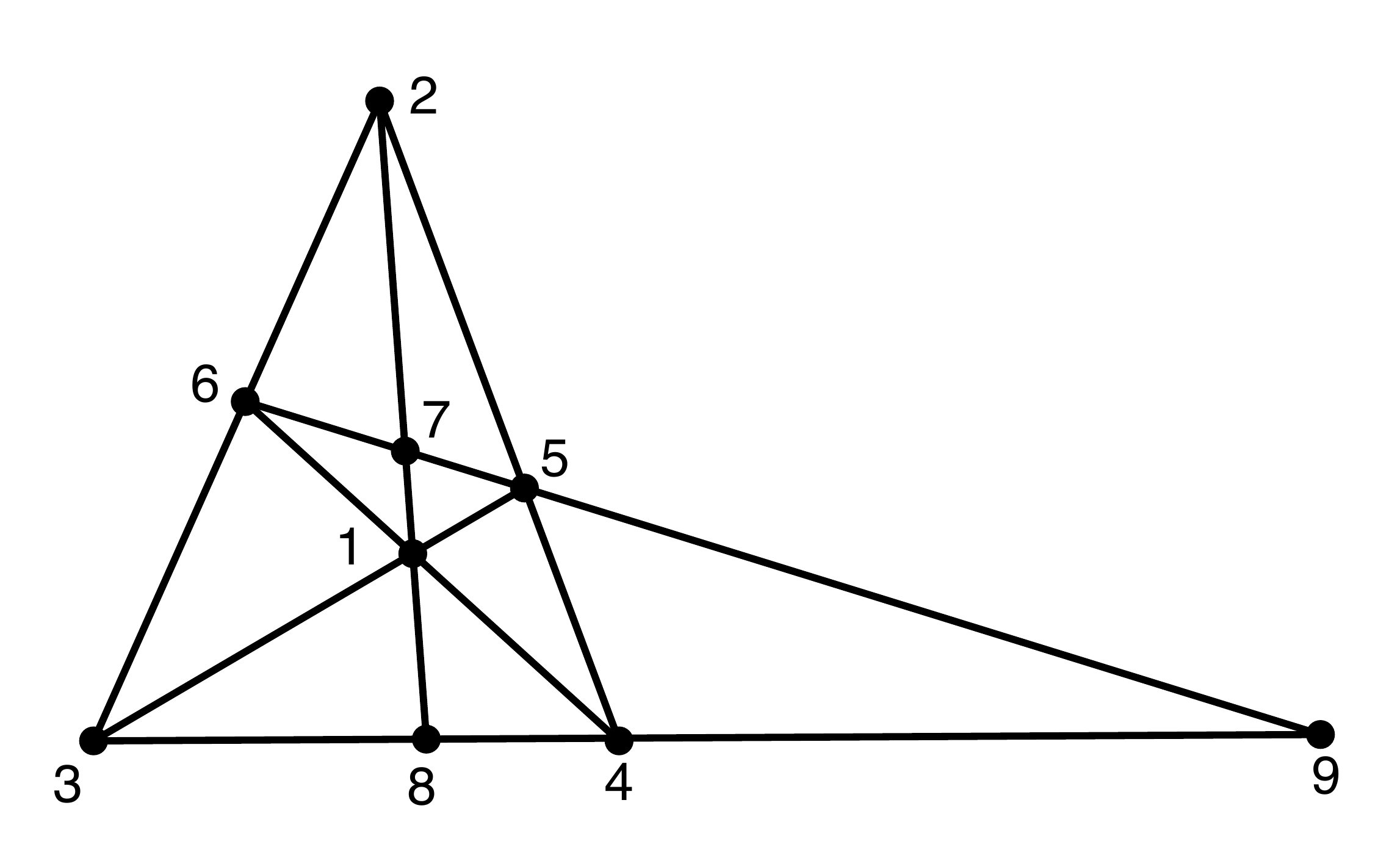} 
\caption{Graphic representation of $B_3$-matroid: lines correspond to rank-2 flats}
\end{center}
\end{figure}

A matroid is called {\it simple} if it does not contain one- or two-element circuits. 
For vector matroids this means that no two vectors are proportional.

Number of matroids up to isomorphism grows very rapidly with $n = |X|.$ The following table summarises the results for rank $3$-matroids for small $n$ (see \cite{YM}).

\medskip

\begin{flushleft}
\begin{tabular}{|c|c|c|c|c|c|c|c|c|c|c|}
\hline 
$n$ & 3 & 4 & 5 & 6 & 7 & 8 & 9 & 10 & 11 & 12\tabularnewline
\hline 
\hline 
all matroids & 1 & 4 & 13 & 38 & 108 & 325 & 1275 & 10037 & 298491 & 31899134\tabularnewline
\hline 
simple matroids & 1 & 2 & 4 & 9 & 23 & 68 & 383 & 5249 & 232928 & 28872972\tabularnewline
\hline 
\end{tabular}
\par\end{flushleft}

\medskip

Vector matroids build the class of {\it realisable matroids.}
The problem of finding a criterion for realisability is known to be $NP$-hard \cite{R-GZ}. 

Let $M$ be a rank $r$ vector matroid. We say that matroid $M$ is {\it projectively rigid} if the space of all its rank $r$ vector realisations
$$\mathcal R(M)=\{A: M = M[A]\}/\sim$$ modulo projective equivalence is discrete and {\it strongly projectively rigid} if it consists of only one point (which means that modulo projective equivalence $M$ has a unique vector realisation).

Let $G$ be a finite Coxeter group, which is a finite group generated by the hyperplane reflections in a Euclidean space.
We say that matroid $M$ is of {\it Coxeter type} if it describes the vector configuration of the normals 
to the corresponding reflection hyperplanes (one for each hyperplane) for such a group.
For rank three Coxeter matroids we have the following result.

\begin{thm}
\label{rigid}
The matroids of Coxeter types $A_3$ and $B_3$ are strongly projectively rigid.
The matroid of type $H_3$ is projectively rigid with precisely two projectively non-equivalent vector realisations.
\end{thm}

\begin{proof}
Let us prove this first for $B_3$ case. Since the images $a_1, a_2, a_3, a_4$ of the elements 1,2,3 and 4 in the projective plane form a projective basis it is enough to prove that the remaining $a_5, a_6, a_7, a_8, a_9$ can be constructed uniquely. From the matroid structure we can see that $x_5$ must be an intersection point of the lines (2-flats) $a_1a_3$ and $a_2a_4.$ We denote this as
$$a_5=(a_1a_3) \wedge (a_2a_4)$$
using the general lattice theory notation. Similarly we have 
$$a_6=(a_2a_3) \wedge (a_1a_4), \quad a_7=(a_1a_2) \wedge (a_5a_6),$$
$$a_8=(a_1a_2) \wedge (a_3a_4), \quad a_9=(a_3a_4) \wedge (a_5a_6).$$
Similarly one can prove the rigidity in $A_3$ case (see Fig. 2). In both these cases the space of realisations modulo projective equivalence consists of only one point.

\begin{figure} [H]
\begin{center}
\includegraphics[scale=0.3]{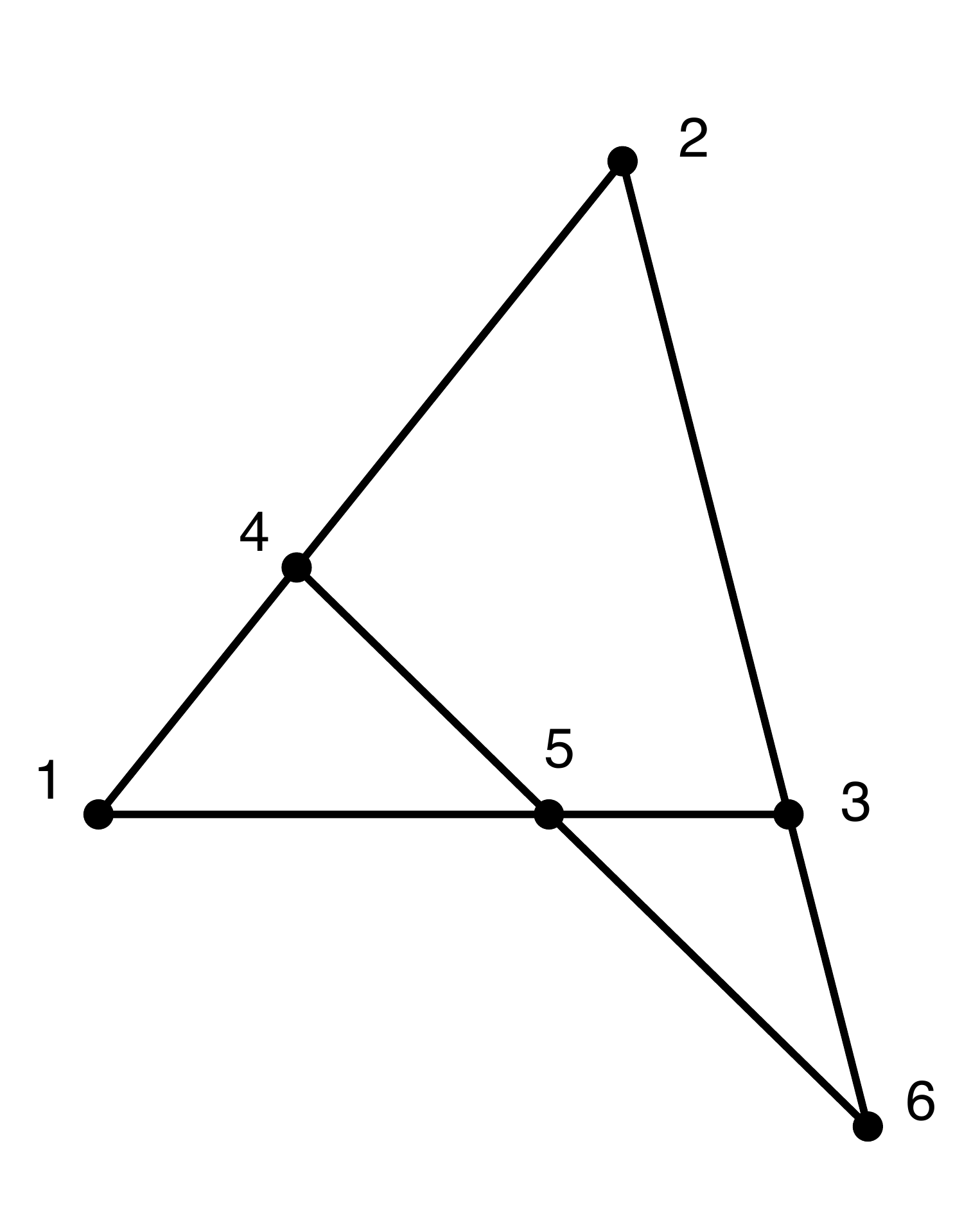} 
\caption{ Graphic representation of $A_3$-matroid}
\end{center}
\end{figure}

The $H_3$ case is more interesting.  Fig. 3 shows the graphic representation of the system $H_{3}$ in
the real projective plane $\mathbb{R}P^{2}$. 

Recall that on the projective line $\mathbb{R}P^{1}$ any three
points can be mapped into any other three via the action of the group
$PGL(2,\mathbb{R}).$ For four distinct points $p_{1},p_{2},p_{3},p_{4}$
on the projective line $\mathbb{R}P^{1}$ with homogeneous coordinates
$[x_{i},y_{i}]$ there is a projective invariant, namely {\it cross-ratio}
defined as 
\[
(p_{1},p_{2};p_{3},p_{4})=\frac{(x_{1}y_{3}-x_{3}y_{1})}{(x_{1}y_{4}-x_{4}y_{1})}\frac{(x_{2}y_{4}-x_{4}y_{2})}{(x_{2}y_{3}-x_{3}y_{2})}.
\]
If none of the $y_{i}$ is zero the cross-ratio can be expressed in
terms of the ratios $z_{i}=\frac{x_{i}}{y_{i}}$ as follows:
\[
(z_{1},z_{2};z_{3},z_{4})=\frac{(z_{1}-z_{3})}{(z_{1}-z_{4})}\frac{(z_{2}-z_{4})}{(z_{2}-z_{3})}.
\]
Since any projection from a point in the projective plane
preserves the cross-ratio of four points we have the equalities
$$(a_{6},a_{5};a_{9},a_{3})=(a_{4},a_{7};a_{10},a_{3})=(a_{5},a_{6};a_{8},a_{3}),$$
$$(a_{6},a_{5};a_{9},a_{3})=(a_{7},a_{11};a_{10},a_{3})=(a_{5},a_{8};a_{9},a_{3}).$$

Using elementary manipulations with cross-ratios one can show that that these equalities imply 
that $x=(a_{6},a_{5};a_{9},a_{3})$ satisfies the equation
$$x^{2}-x-1=0$$ with two solutions $x_{1}=\frac{1+\sqrt{5}}{2}$
and $x_{2}=\frac{1-\sqrt{5}}{2}$.

If we fix the positions of the four points $a_4, a_5, a_6, a_7$ forming a projective basis in $\mathbb{R}P^{2}$ 
we can first reconstruct 
$$a_{1}=(a_{5}a_{4})\wedge(a_{6}a_{7}),\quad a_{2}=(a_{5}a_{7})\wedge(a_{6}a_{4}), \quad a_{3}=(a_{5}a_{6})\wedge(a_{7}a_{4}).$$
Then using the knowledge of $x=(a_{6},a_{5};a_{9},a_{3})$ we can reconstruct $a_9$
and all the remaining points as
$$ a_{14}=(a_{2}a_{9})\wedge(a_{5}a_{4}), \quad a_{12}=(a_{7}a_{6})\wedge(a_{2}a_{9}), \quad a_{10}=(a_{2}a_{9})\wedge(a_{3}a_{4}),$$
$$a_{13}=(a_{9}a_{4})\wedge(a_{6}a_{7}), \quad a_{8}=(a_{2}a_{13})\wedge(a_{3}a_{6}), \quad a_{15}=(a_{2}a_{13})\wedge(a_{5}a_{4}),$$
$$ a_{11}=(a_{2}a_{8})\wedge(a_{3}a_{4}).$$
Thus we have shown that modulo projective group we have only two different vector realisations of matroid $H_3.$
\end{proof}

\begin{figure} [H]
\begin{center}
\includegraphics[scale=0.3]{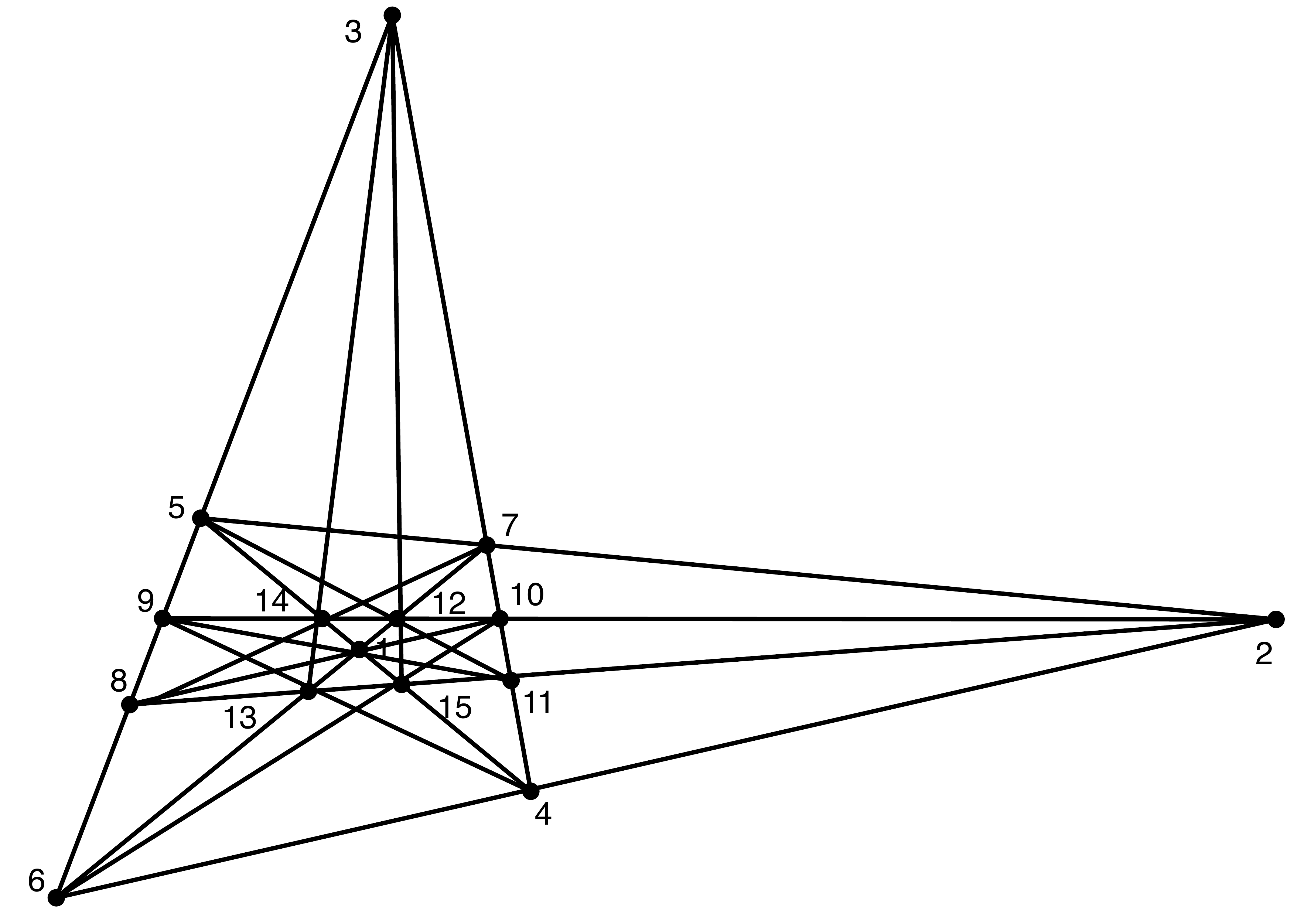} 
\caption{ Graphic representation of $H_3$-matroid}
\end{center}
\end{figure}

{\it Remark.} The existence of two projectively non-equivalent realisations is related to the existence of a symmetry of matroid $M(H_3)$, which can not be realised geometrically, see  \cite{EG}. These two realisations are related by re-ordering of the vectors and thus give rise to the equivalent $\vee$-systems.

\section{Classification Problem for $\vee$-Systems of Given Matroidal Type}

For any $\vee$-system $\mathcal A \subset V^*$ one can consider the corresponding matroid $M(\mathcal A)$, which encodes a combinatorial structure of $\mathcal A.$
Conversely, having a matroid $M$ one can look for {\it $\vee$-system realisations} $\mathcal A$ of $M$ with given combinatorial structure $M(\mathcal A)=M.$

Let $\mathcal R^{\vee}(M)$ be the set of all such realisations modulo group $G=GL(V^*)$ of linear automorphisms of $V^*.$ 

If vector matroid $M=M(A)$ is strongly projectively rigid then all its vector realisations modulo $G$ have the form $A'=AD$, or in terms of the columns $a_i, \, i=1,\dots,n$ of $A$,
$$a_i'=x_ia_i, \, \, i=1,\dots,n$$
with arbitrary non-zero parameters $x_i.$ The $\vee$-conditions form a system of nonlinear algebraic relations on the parameters $x_i \in \mathbb R\setminus 0$  and define $\mathcal R^{\vee}(M)$ as an open set of a real algebraic variety.

For a generic vector matroid this set is actually empty. For example, for $n$ vectors $a_i$ in $\mathbb R^3$ in general position the $\vee$-conditions imply that these vectors must be pairwise orthogonal, which is impossible if $n>3.$

In the case when the space $\mathcal R^{\vee}(M)$ is known to be non-empty (for example, for all vector matroids $M$ of Coxeter type) we have the question of how to describe this space effectively.

For the case of matroid of Coxeter type $A_{3}$ the answer is known \cite{CV}. The positive roots of $A_3$ system are $e_i-e_j, \,\, 1\leq i<j\leq 4$, where $e_i,\, i=1,\dots, 4$ is an orthonormal basis in $\mathbb R^4.$

Since matroid $A_3$ is strongly projectively rigid it is enough to consider the system 
\begin{equation}
\mathcal{A}=\left\{ \mu_{ij}\left(e_{i}-e_{j}\right),1\leq i<j\leq4\right\}. \label{Adef}
\end{equation}

\begin{thm}
\label{defA3}\cite{CV} 
The system (\ref{Adef}) satisfies the $\vee$-conditions if and only if the parameters satisfy the relations 
\[
\mu_{12}\mu_{34}=\mu_{13}\mu_{24}=\mu_{14}\mu_{23}.
\]
\medskip{}
All the corresponding $\vee$-systems can be parametrized as 
\[
A_{3}(c)=\left\{\sqrt{c_{i}c_{j}}(e_{i}-e_{j}),1\leq i<j\leq4\right\} ,
\]
with arbitrary positive real 
$c_{1},\ldots, c_{4}.$ \medskip{}
\end{thm}

Without loss of generality, we may choose $c_{4}=1$ and consider
the restriction of the system onto the hyperplane $x_{4}=0.$  
This gives the following parametrisation of the space $\mathcal R^{\vee}(M(A_3))$ by positive 
real $c_{1},c_2, c_{3}$ as 
\[
A_{3}(c)=\begin{cases}
\sqrt{c_{i}c_{j}}(e_{i}-e_{j}), & 1\leq i<j\leq3\\
\sqrt{c_{i}}e_{i}, & i=1,2,3.
\end{cases}
\]

Consider now the case $B_3,$
corresponding to the following configuration of vectors in $\mathbb{R}^{3}$
\[
B_{3}=\begin{cases}
e_{i}\pm e_{j}, & 1\leq i<j\leq3,\\
e_{i} & i=1,\ldots, 3.
\end{cases}
\]

The following 4-parametric family of $\vee$-systems of $B_3$-type was found in \cite{CV}:
\begin{equation}
\label{CV}
\mathcal{B}\left(c, \gamma\right)=\begin{cases}
\sqrt{c_{i}c_{j}}\left(e_{i}\pm e_{j}\right), & 1\leq j<i\leq3\\
\sqrt{2c_{i}(c_{i}+\gamma)}e_{i}, & 1\leq i\leq3
\end{cases}
\end{equation}
with arbitrary positive $c_{1},c_{2},c_{3}$ and $\gamma$ such that $c_i+\gamma>0$ for all $i=1,2,3.$

\begin{thm}
Formula (\ref{CV}) gives all rank 3 $\vee$-systems of matroid type $B_3.$
\end{thm}

\begin{proof}
The proof is by direct computations, but we present here the details
to show the algebraic nature of $\vee$-conditions in this example.

Since $B_3$ matroid is strongly projectively rigid, we can assume that the corresponding $\vee$-system has the form
\[
\mathcal{A}=\begin{cases}
\alpha_{ij}\left(e_{i}+e_{j}\right), & 1 \leq i<j\leq 3\\
\tilde{\alpha}_{ij}\left(e_{i}-e_{j}\right), & 1 \leq j<i\leq 3\\
\beta_{i}e_{i}, & 1 \leq i\leq 3,
\end{cases}
\]
\medskip{}
where all the parameters can be assumed without loss of generality to be positive.

To write down all $\vee$-conditions consider all two-dimensional planes 
containing at least two vectors $v_{1},v_{2}\in\mathcal{A}.$

There are 3 different types of such planes $\Pi$:

\begin{enumerate}
\item $<e_{1},e_{2}\pm e_{3}>, <e_{2},e_{1}\pm e_{3}>, <e_{3},e_{1}\pm e_{2}>,$

\item $<e_{1},e_{2},e_{1}\pm e_{2}>, <e_{1},e_{3},e_{1}\pm e_{3}>, <e_{2},e_{3},e_{2}\pm e_{3}>,$

\item $<e_{1}-e_{2},e_{2}-e_{3},e_{1}-e_{3}>,$ $<e_{1}-e_{2},e_{2}+e_{3},e_{1}+e_{3}>,$ \\
$<e_{2}-e_{3},e_{1}+e_{3},e_{1}+e_{2}>,<e_{1}-e_{3},e_{2}+e_{3},e_{1}+e_{2}>.$ 
\end{enumerate}

The corresponding form $G$ has the matrix
\[
G=\left(\begin{array}{ccc}
\alpha_{13}^{2}+\alpha_{12}^{2}+\tilde{\alpha}_{13}^{2}+\tilde{\alpha}_{12}^{2}+\beta_{1}^{2} & \alpha_{12}^{2}-\tilde{\alpha}_{12}^{2} & \alpha_{13}^{2}-\tilde{\alpha}_{13}^{2}\\
\alpha_{12}^{2}-\tilde{\alpha}_{12}^{2} & \alpha_{23}^{2}+\alpha_{12}^{2}+\tilde{\alpha}_{23}^{2}+\tilde{\alpha}_{12}^{2}+\beta_{2}^{2} & \alpha_{23}^{2}-\tilde{\alpha}_{23}^{2}\\
\alpha_{13}^{2}-\tilde{\alpha}_{13}^{2} & \alpha_{23}^{2}-\tilde{\alpha}_{23}^{2} & \alpha_{23}^{2}+\alpha_{13}^{2}+\tilde{\alpha}_{23}^{2}+\tilde{\alpha}_{13}^{2}+\beta_{3}^{2}
\end{array}\right)
\]

In case 1. the $\vee$-conditions are just the orthogonality conditions $G(\alpha^{\vee}, \beta^{\vee})=0$ 
for the corresponding two covectors $\alpha$ and $\beta$ in the plane $\Pi.$ 
We obtain the system
\[
\begin{cases}
2(\tilde{\alpha}_{23}^{2}\tilde{\alpha}_{13}^{2}+\tilde{\alpha}_{23}^{2}\tilde{\alpha}_{12}^{2}+\tilde{\alpha}_{13}^{2}\tilde{\alpha}_{12}^{2}-\alpha_{13}^{2}\alpha_{12}^{2}-\alpha_{13}^{2}\tilde{\alpha}_{23}^{2}-\alpha_{12}^{2}\tilde{\alpha}_{23}^{2})-\alpha_{13}^{2}\beta_{2}^{2}+\tilde{\alpha}_{13}^{2}\beta_{2}^{2}-\alpha_{12}^{2}\beta_{3}^{2}+\tilde{\alpha}_{12}^{2}\beta_{3}^{2}=0\\
2(\tilde{\alpha}_{23}^{2}\tilde{\alpha}_{13}^{2}+\tilde{\alpha}_{23}^{2}\tilde{\alpha}_{12}^{2}+\tilde{\alpha}_{13}^{2}\tilde{\alpha}_{12}^{2}-\alpha_{23}^{2}\alpha_{12}^{2}-\alpha_{23}^{2}\tilde{\alpha}_{13}^{2}-\alpha_{12}^{2}\tilde{\alpha}_{13}^{2})-\alpha_{23}^{2}\beta_{1}^{2}+\tilde{\alpha}_{23}^{2}\beta_{1}^{2}-\alpha_{12}^{2}\beta_{3}^{2}+\tilde{\alpha}_{12}^{2}\beta_{3}^{2}=0\\
2(\tilde{\alpha}_{23}^{2}\tilde{\alpha}_{13}^{2}-\alpha_{23}^{2}\alpha_{13}^{2}-\alpha_{23}^{2}\tilde{\alpha}_{12}^{2}-\alpha_{13}^{2}\tilde{\alpha}_{12}^{2}+\tilde{\alpha}_{23}^{2}\tilde{\alpha}_{12}^{2}+\tilde{\alpha}_{13}^{2}\tilde{\alpha}_{12}^{2})-\alpha_{23}^{2}\beta_{1}^{2}+\tilde{\alpha}_{23}^{2}\beta_{1}^{2}-\alpha_{13}^{2}\beta_{2}^{2}+\tilde{\alpha}_{13}^{2}\beta_{2}^{2}=0\\
2(\alpha_{23}^{2}\alpha_{12}^{2}+\alpha_{23}^{2}\tilde{\alpha}_{13}^{2}+\alpha_{12}^{2}\tilde{\alpha}_{13}^{2}-\alpha_{23}^{2}\alpha_{13}^{2}-\alpha_{23}^{2}\tilde{\alpha}_{12}^{2}-\alpha_{13}^{2}\tilde{\alpha}_{12}^{2})-\alpha_{13}^{2}\beta_{2}^{2}+\tilde{\alpha}_{13}^{2}\beta_{2}^{2}+\alpha_{12}^{2}\beta_{3}^{2}-\tilde{\alpha}_{12}^{2}\beta_{3}^{2}=0\\
2(\alpha_{13}^{2}\alpha_{12}^{2}+\alpha_{13}^{2}\tilde{\alpha}_{23}^{2}+\alpha_{12}^{2}\tilde{\alpha}_{23}^{2}-\alpha_{23}^{2}\tilde{\alpha}_{12}^{2}-\alpha_{13}^{2}\tilde{\alpha}_{12}^{2}-\alpha_{23}^{2}\alpha_{13}^{2})-\alpha_{23}^{2}\beta_{1}^{2}+\tilde{\alpha}_{23}^{2}\beta_{1}^{2}+\alpha_{12}^{2}\beta_{3}^{2}-\tilde{\alpha}_{12}^{2}\beta_{3}^{2}=0\\
2(\alpha_{13}^{2}\alpha_{12}^{2}+\alpha_{13}^{2}\tilde{\alpha}_{23}^{2}+\alpha_{12}^{2}\tilde{\alpha}_{23}^{2}-\alpha_{23}^{2}\tilde{\alpha}_{13}^{2}-\alpha_{12}^{2}\tilde{\alpha}_{13}^{2}-\alpha_{23}^{2}\alpha_{12}^{2})-\alpha_{23}^{2}\beta_{1}^{2}+\tilde{\alpha}_{23}^{2}\beta_{1}^{2}+\alpha_{13}^{2}\beta_{2}^{2}-\tilde{\alpha}_{13}^{2}\beta_{2}^{2}=0,
\end{cases}
\]
which can be reduced to
\[
\begin{cases}
(-\alpha_{12}^{2}+\tilde{\alpha}_{12}^{2})(\alpha_{23}^{2}+\alpha_{13}^{2}+\tilde{\alpha}_{23}^{2}+\tilde{\alpha}_{13}^{2}+\beta_{3}^{2}-\frac{(\alpha_{13}^{2}-\tilde{\alpha}_{13}^{2})^{2}}{(\alpha_{13}^{2}+\alpha_{12}^{2}+\tilde{\alpha}_{13}^{2}+\tilde{\alpha}_{12}^{2}+\beta_{1}^{2})})=0\\
(-\alpha_{13}^{2}+\tilde{\alpha}_{13}^{2})(\alpha_{23}^{2}+\alpha_{12}^{2}+\tilde{\alpha}_{23}^{2}+\tilde{\alpha}_{12}^{2}+\beta_{2}^{2}-\frac{(\alpha_{12}^{2}-\tilde{\alpha}_{12}^{2})^{2}}{(\alpha_{13}^{2}+\alpha_{12}^{2}+\tilde{\alpha}_{13}^{2}+\tilde{\alpha}_{12}^{2}+\beta_{1}^{2})})=0\\
(-\alpha_{23}^{2}+\tilde{\alpha}_{23}^{2})(\alpha_{12}^{2}+\alpha_{13}^{2}+\tilde{\alpha}_{12}^{2}+\tilde{\alpha}_{13}^{2}+\beta_{1}^{2}-\frac{(\alpha_{13}^{2}-\tilde{\alpha}_{13}^{2})^{2}}{(\alpha_{13}^{2}+\alpha_{23}^{2}+\tilde{\alpha}_{13}^{2}+\tilde{\alpha}_{23}^{2}+\beta_{3}^{2})})=0.
\end{cases}
\]
Note that the second factors in all equations are ratios of principal minors of matrix $G$ and thus must be positive, since the form $G$ is positive definite. This implies that $\alpha_{ij}=\tilde \alpha_{ij}$, which reduces the matrix $G$ to
\[
G=\left(\begin{array}{ccc}
2(\alpha_{13}^{2}+\alpha_{12}^{2})+\beta_{1}^{2} & 0 & 0\\
0 & 2(\alpha_{23}^{2}+\alpha_{12}^{2})+\beta_{2}^{2} & 0\\
0 & 0 & 2(\alpha_{23}^{2}+\alpha_{13}^{2})+\beta_{3}^{2}
\end{array}\right).
\]

In cases 2. and 3. we fix for each plane $\Pi$ a basis $v_{1},v_{2}\in \mathcal A\cap\Pi.$
The corresponding dual plane $\Pi^{\vee}$ is spanned by $v_{1}^{\vee}$
and $v_{2}^{\vee}$ and the $\vee$-condition implies
the proportionality of the restrictions of the forms $G$ and $G_{\Pi}$
onto $\Pi^{\vee}$. In our case this proportionality turns out
to be equivalent to the following system of equations:
$$
\begin{cases}
\frac{2\alpha_{12}^{2}}{2(\alpha_{23}^{2}+\alpha_{12}^{2})+\beta_{2}^{2}}-\frac{2\alpha_{13}^{2}}{2(\alpha_{23}^{2}+\alpha_{13}^{2})+\beta_{3}^{2}}=0\\
\frac{2\alpha_{12}^{2}}{2(\alpha_{13}^{2}+\alpha_{12}^{2})+\beta_{1}^{2}}-\frac{2\alpha_{23}^{2}}{2(\alpha_{23}^{2}+\alpha_{13}^{2})+\beta_{3}^{2}}=0\\
\frac{2\alpha_{13}^{2}}{2(\alpha_{13}^{2}+\alpha_{12}^{2})+\beta_{1}^{2}}-\frac{2\alpha_{23}^{2}}{2(\alpha_{23}^{2}+\alpha_{12}^{2})+\beta_{2}^{2}}=0.
\end{cases}
$$
Introducing new parameters  $c_{i}, \, i=1,2,3$ and $\gamma$
by
$$
c_{i}:=\frac{\alpha_{ij}\alpha_{ik}}{\alpha_{jk}}, \quad 
\gamma:=\frac{\beta_{3}^{2}-2c_{3}^{2}}{2c_{3}}. 
$$
we can see that these relations imply
$$\alpha_{ij}^{2}=c_{i}c_{j}, \quad \beta_{i}^{2}=2c_{i}(c_{i}+\gamma),$$
which leads to the parametrisation (\ref{CV}).
\end{proof}

For larger matroids  the direct analysis of the $\vee$-conditions is very difficult, so we consider a simpler problem about infinitesimal deformations of $\vee$-systems.

\section{Deformations of $\vee$-Systems}

Let $\mathcal A=\{ \alpha\} \subset V^*$ be a $\vee$-system realisation of matroid $M.$
Consider its smooth {\it scaling deformation} $\mathcal A(t)$ of the form
\begin{equation}
\label{def}
\mathcal A(t)=\{\alpha_t\}, \quad \alpha_t=\mu_{\alpha}(t)\alpha, \,\, \mu_{\alpha}(0)=1.
\end{equation}
For projectively rigid matroids $M$ one can always reduce any deformation to such a form.

Let $\xi_{\alpha}=\dot \mu_{\alpha}(0)$.
We are going to derive the conditions on $\xi_{\alpha}$, 
which can be considered as {\it linearised $\vee$-conditions} for such deformations.

Let $$G_t(x,y):=G_{\mathcal A(t)}(x,y)=\sum_{\a\in\mathcal{A}}\a_t(x)\a_t(y)$$ with $G_0=G=G_\mathcal A(t)$
and consider its derivative
$$
\dot G_t(x,y)=\sum_{\a\in\mathcal{A}}\dot\a_t(x)\a_t(y)+\sum_{\a\in\mathcal{A}}\a_t(x)\dot \a_t(y),
$$
which at $t=0$ gives $\dot G_0(x,y)=2X,$ where 
$$X=\sum_{\a\in\mathcal{A}}\xi_{\alpha} \a(x)\a(y).$$

Consider now the $\vee$-conditions.

For any two-dimensional plane containing only two covectors we have
$$G_t(\alpha_t^{\vee}, \beta_t^{\vee})=0.$$
Differentiating it in $t$ we have
\begin{equation}
\label{1}
\dot G_t(\alpha_t^{\vee}, \beta_t^{\vee})+G(\dot\alpha_t^{\vee}, \beta_t^{\vee})+G(\alpha^{\vee}, \dot\beta_t^{\vee})=0,
\end{equation}
where here and below by $\dot\alpha_t^{\vee}$ we mean $\frac{d}{dt}(\alpha_t^{\vee}).$

To find $G(\dot\alpha_t^{\vee}, \beta_t^{\vee})$ note that by definition of $\alpha_t^{\vee}$
$G_t(\alpha_t^{\vee},v)=\alpha_t(v)$ for any fixed vector $v \in V.$
Differentiating this with respect to $t$ we have
$$\dot G_t(\alpha_t^{\vee},v)+G_t(\dot \alpha_t^{\vee},v)=\dot \alpha_t(v)$$
which for $t=0$ gives
$$2X(\alpha_t^{\vee},v)+G(\dot \alpha_t^{\vee},v)=\xi_{\alpha}\alpha(v).$$
Thus we have
$$
G(\dot \alpha_0^{\vee},v)=\xi_{\alpha}\alpha(v)-2X(\alpha^{\vee},v).
$$
and thus
$$
G(\dot \alpha_0^{\vee},\beta^{\vee})=\xi_{\alpha}\alpha(\beta^{\vee})-2X(\alpha^{\vee},\beta^{\vee})=-2X(\alpha^{\vee},\beta^{\vee})
$$
since $\alpha(\beta^{\vee})=G(\alpha^{\vee}, \beta^{\vee})=0$ by the $\vee$-conditions.

Substituting this into (\ref{1}) we have the first linearised $\vee$-condition: 
for $\alpha,\beta$ being the only two covectors in a plane $\Pi$ we have
\begin{equation}
\label{1X}
X(\alpha^{\vee}, \beta^{\vee})=0.
\end{equation}

Let now $\Pi$$ $ be a two-dimensional plane containing
more than two covectors from $\mathcal A$ (and hence from $\mathcal A_t.$ Then from the $\vee$-conditions 
there exists $\nu = \nu(\Pi) \in \mathbb R$ such that for any $\alpha \in \Pi \cap \mathcal A, v \in V$ 
we have
\begin{equation}
\label{vee1}
G^{\Pi}(\alpha^{\vee}, v)=\nu G(\alpha^{\vee}, v),
\end{equation}
where
$G^{\Pi}(x,y)=G^{\Pi}_{\mathcal A}(x,y)=\sum_{\a\in \Pi\cap\mathcal{A}}\a(x)\a(y)$ (see \cite{FV2}).
Now assuming that $\mathcal A$ depends on $t$ as above and differentiating with respect to $t$ at $t=0$ we have as before for any $\alpha, \beta \in \mathcal A \cap \Pi$
$$
\dot G^{\Pi}(\alpha^{\vee}, \beta^{\vee})+G^{\Pi}(\dot\alpha^{\vee}, \beta^{\vee})+G^{\Pi}(\alpha^{\vee}, \dot\beta^{\vee})=\dot\nu G(\alpha^{\vee}, \beta^{\vee})+\nu \dot G(\alpha^{\vee}, \beta^{\vee})\\
+\nu G(\dot\alpha^{\vee}, \beta^{\vee})+\nu G(\alpha^{\vee}, \dot\beta^{\vee}).
$$
But from (\ref{vee1}) we have $G^{\Pi}(\dot\alpha^{\vee}, \beta^{\vee})=\nu G(\dot\alpha^{\vee}, \beta^{\vee})$ and 
$G^{\Pi}(\alpha^{\vee}, \dot\beta^{\vee})=\nu G(\alpha^{\vee}, \dot\beta^{\vee})$. Since $\dot G^{\Pi}=2X^{\Pi},$
where
$$
X^{\Pi}(x,y)=\sum_{\a\in \Pi\cap\mathcal{A}}\xi_{\alpha}\a(x)\a(y),
$$
we have
$$2X^{\Pi}(\alpha^{\vee}, \beta^{\vee})=\dot \nu G(\alpha^{\vee}, \beta^{\vee}) +2\nu X(\alpha^{\vee}, \beta^{\vee}),$$
or, eventually
\begin{equation}
\label{2}
2(X^{\Pi}-\nu X)((\alpha^{\vee}, \beta^{\vee})=\dot \nu G(\alpha^{\vee}, \beta^{\vee}).
\end{equation}
Since this is true for all $\alpha,\beta \in \Pi \cap \mathcal A$ we have the second linearised $\vee$-condition:
for any plane $\Pi$ containing more than two covectors from $\mathcal A$ we have
\begin{equation}
\label{2X}
X^{\Pi}-\nu X \sim G\mid_{\Pi^{\vee}},
\end{equation}
where the sign $\sim$ means proportionality.

Thus we have proved
\begin{thm}
The deformations of $\vee$-systems of the form (\ref{def}) are described by the linear $\vee$-conditions (\ref{1X}), (\ref{2X}).
For projectively rigid matroidal types this describes all infinitesimal deformations of a given $\vee$-system.
\end{thm}

Case by case check of the $\vee$-systems from the Appendix leads to the following
\begin{thm}
All rank three vector matroids corresponding to known irreducible 3D $\vee$-systems are projectively rigid.
The $H_3$ matroid is the only one, which is not strongly projectively rigid.
\end{thm}

Let us show that the largest known case $(H_4, A_1)$ is strongly projectively rigid. 
We will use the labelling of the points shown at the last figure of the paper.
Fix the positions of the four points $6, 25, 27, 30$ forming a projective basis in $\mathbb{R}P^{2}.$ 
After this all the remaining points can be reconstructed uniquely as follows:
$$31=(25,30)\wedge(6,27), \quad 29=(25,27)\wedge(6,30),\quad 12=(6,25)\wedge(27,30),$$
$$9=(30,27)\wedge(29,31), \quad 17=(25,21)\wedge(6,30), \quad 28=(17,27)\wedge(12,25),$$
$$4=(28,31)\wedge(12,30), \quad 24=(17,27)\wedge(25,30), \quad 3=(25,30)\wedge(28,29),$$
$$23=(3,28)\wedge(27,31), \quad 11=(27,31)\wedge(28,30), \quad 19=(3,9)\wedge(4,8),$$
$$1=(25,30)\wedge(11,12), \quad 16=(7,11)\wedge(1,19), \quad 20=(16,25)\wedge(4,8),$$
$$7=(1,4)\wedge(25,27), \quad 21=(25,17)\wedge(7,11), \quad 10=(7,31)\wedge(4,25),$$
$$26=(7,31)\wedge(4,16), \quad 14=(4,16)\wedge(11,31), \quad18=(21,31)\wedge(28,25),$$
$$22=(7,31)\wedge(3,28), \quad 15=(7,31)\wedge(24,28), \quad 2=(15,19)\wedge(10,11),$$
$$5=(24,28)\wedge(21,31), \quad 13=(7,25)\wedge(12,24), \quad 8=(1,12)\wedge(10,16).$$

A direct computation shows that in case of the classical systems
$A_{3}$ and $B_{3}$ the linear system (\ref{1}),(\ref{2}) has corank four
in agreement with the results of the previous section.

The analysis of the linearised $\vee$-conditions
for the families $D_{3}(t,s)$, $F_{3}(t)$, $G_{3}(t)$ and $(AB_{4}(t),A_{1})_{1,2}$
shows that these families of $\vee$-systems can not be extended.

Consider, for example, the family of $\vee$-systems $D_3(t,s)$ from \cite{FV2} with
\[A=
\left(\begin{array}{ccccccc}
1 & 1 & 1 & 1 & 0 & 0 & \sqrt{2}\sqrt{s+t-1}\\
1 & -1 & -1 & 1 & \sqrt{2}\sqrt{\frac{s-t+1}{t}} & 0 & 0\\
1 & -1 & 1 & -1 & 0 & \sqrt{2}\sqrt{\frac{-s+t+1}{s}} & 0
\end{array}\right)
\]
with real parameters $s,t$ such that $|s-t|<1, \, s+t>1.$
Matrices $G$ and $X$ have the form
\[
G=\left(\begin{array}{ccc}
2(s+t+1) & 0 & 0\\
0 & \frac{2(s+t+1)}{t} & 0\\
0 & 0 & \frac{2(s+t+1)}{s}
\end{array}\right)
\]
$$
\resizebox{1.11\hsize}{!}
{$
X=\left(\begin{array}{ccc}
\xi_{1}+\xi_{2}+\xi_{3}+\xi_{4}+2\xi_{7}(s+t-1) & \xi_{1}-\xi_{2}-\xi_{3}+\xi_{4} & \xi_{1}-\xi_{2}+\xi_{3}-\xi_{4}\\
\xi_{1}-\xi_{2}-\xi_{3}+\xi_{4} & \xi_{1}+\xi_{2}+\xi_{3}+\xi_{4}+\frac{2(s+1)}{t}\xi_{5}-2\xi_{5} & \xi_{1}+\xi_{2}-\xi_{3}-\xi_{4}\\
\xi_{1}-\xi_{2}+\xi_{3}-\xi_{4} & \xi_{1}+\xi_{2}-\xi_{3}-\xi_{4} & \xi_{1}+\xi_{2}+\xi_{3}+\xi_{4}+\frac{2(-s+{t}+1)}{s}\xi_{6}
\end{array}\right)$}.
$$

For the three covectors $\alpha_{5},\alpha_{6},\alpha_{7}$  the first linearised $\vee$-conditions $X(\alpha_{i}^{\vee}, \alpha_{j}^{\vee})=0$, $i,j=5,6,7$ are equivalent to
\[
\xi_{1}
+\xi_{2}
-\xi_{3}
-\xi_{4}
=0,
\]
\[
\xi_{1}
-\xi_{2}
-\xi_{3}
+\xi_{4}
=0,
\]
\[
\xi_{1}
-\xi_{2}
+\xi_{3}
-\xi_{4}
=0,
\]
which imply that $\xi_1=\xi_2=\xi_3=\xi_4.$

For the planes with more than two covectors we have the linear system
\begin{multline*}
  (s+t) (\xi_{1} (s+t+1)+\xi_{2} (s+t-3)+t (\xi_{3} (2 s+3)+\xi_{4} (2 s-1)-2 (s (\xi_{5}+\xi_{6})+\xi_{5}+\xi_{7})) \\  
  +t^2 (-(\xi_{3}+\xi_{4}-2 \xi_{5}))-s (s (\xi_{3}+\xi_{4}-2 \xi_{6})+\xi_{3}-3 \xi_{4}+2 (\xi_{6}+\xi_{7}))+2 \xi_{7})=0,
\end{multline*}
\begin{multline*}
(s + t) ((s-1) (\xi_{2} -\xi_{7} + s(\xi_{6} - \xi_{4})) + t(\xi_{2} +\xi_{3} - \xi_{5} - \xi_{7} +s (\xi_{3} + \xi_{4} - \xi_{5} - \xi_{6})) \\ +  t^2(-\xi_{3} + \xi_{5}))=0,
\end{multline*}      
\begin{multline*}
(s+1) t (s (\xi_{1}-\xi_{2}+\xi_{3}+3 \xi_{4}-2 (\xi_{5}+\xi_{6}))+\xi_{1}+3 \xi_{2}+\xi_{3}-\xi_{4}-2 (\xi_{5}+\xi_{7}))+(s+1) \\ 
	 t^2 (\xi_{1}-3 \xi_{3}+2 \xi_{5}) +(s^2-1) (-\xi_{2} (s-1)-\xi_{4} (s-1)+2 \xi_{6} s-2 \xi_{7})=0,
\end{multline*}
\begin{multline*}
(s+1) t (\xi_{2}+s (\xi_{3}+\xi_{4}-\xi_{5}-\xi_{6})+\xi_{3}-\xi_{5}-\xi_{7})+(s^2-1) (\xi_{2}+s (\xi_{6}-\xi_{4})-\xi_{7}) \\
	+(s+1) t^2   (-(\xi_{3}-\xi_{5}))=0,
\end{multline*}
\begin{multline*}
(t+1) (-t (s (\xi_{1}+\xi_{4}-2 (\xi_{5}+\xi_{6}))+\xi_{3} (3 s+2)-2 (\xi_{5}+\xi_{7}))+s (-s (\xi_{1}+2 \xi_{6})+\xi_{3} \\
+\xi_{4} (3 s-1)+2   (\xi_{6}+\xi_{7}))-\xi_{1} s+\xi_{2} s (t-3)+(t-1)^2)+t^2 (\xi_{3}-2 \xi_{5})+\xi_{3}-2 \xi_{7}=0,
\end{multline*}
   \begin{multline*}
(t+1)(t (\xi_{2}+s (\xi_{3}+\xi_{4}-\xi_{5}-\xi_{6})+\xi_{3}-\xi_{5}-\xi_{7})+(s-1) (\xi_{2}+s (\xi_{6}-\xi_{4}) \\
-\xi_{7})+t^2   (\xi_{5}-\xi_{3})=0,
\end{multline*}
\begin{multline*}
(t+1)(\xi_{1} (s (t-3)+(t-1)^2)-t (s (\xi_{2}+\xi_{3}-2 (\xi_{5}+\xi_{6}))+\xi_{4} (3 s+2)-2 (\xi_{5}+\xi_{7})) \\
	-s (s (\xi_{2}-3   \xi_{3}+2 \xi_{6})+\xi_{2}+\xi_{3}-\xi_{4}-2 (\xi_{6}+\xi_{7}))+t^2 (\xi_{4}-2 \xi_{5})+\xi_{4}-2 \xi_{7})=0,
\end{multline*}
 \begin{multline*}
(t+1) (t (\xi_{1}+s (\xi_{3}+\xi_{4}-\xi_{5}-\xi_{6})+\xi_{4}-\xi_{5}-\xi_{7})+(s-1) (\xi_{1}+s (\xi_{6}-\xi_{3}) \\
-\xi_{7})+t^2   (\xi_{5}-\xi_{4}))=0,
\end{multline*}
\begin{multline*}
(s+1) t (\xi_{1} (s-3)-s (\xi_{2}+3 \xi_{3}+\xi_{4}-2 (\xi_{5}+\xi_{6}))-\xi_{2}+\xi_{3}-\xi_{4}+2 (\xi_{5}+\xi_{7})) \\
+\left(s^2-1\right) (\xi_{1}(s-1)+\xi_{3} (s-1)-2 \xi_{6} s+2 \xi_{7})+(s+1) t^2 (-(\xi_{2}-3 \xi_{4}+2 \xi_{5}))=0,
\end{multline*}
\begin{multline*}
(s+1) t (\xi_{1}+s (\xi_{3}+\xi_{4}-\xi_{5}-\xi_{6})+\xi_{4}-\xi_{5}-\xi_{7})+(s^2-1) (\xi_{1}+s (\xi_{6}-\xi_{3}) \\
-\xi_{7})+(s+1) t^2   (-(\xi_{4}-\xi_{5}))=0,
\end{multline*}
\begin{multline*}
(s+t) (t (\xi_{1} (2 s+3)+\xi_{2} (2 s-1)+\xi_{4}-2 (s (\xi_{5}+\xi_{6})+\xi_{5}+\xi_{7}))+s (-s (\xi_{1}+\xi_{2}-2 \xi_{6}) \\
+\xi_{4}-2   (\xi_{6}+\xi_{7}))+t^2 (-(\xi_{1}+\xi_{2}-2 \xi_{5}))-\xi_{1} s+3 \xi_{2} s+\xi_{3} (s+t+1)-3 \xi_{4}+2 \xi_{7})=0,
\end{multline*}
 \begin{multline*}
  (s+t) (t (s (\xi_{1}+\xi_{2}-\xi_{5}-\xi_{6})+\xi_{1}+\xi_{4}-\xi_{5}-\xi_{7})+t^2 (\xi_{5}-\xi_{1})+(s-1) (s   (\xi_{6}-\xi_{2}) \\
  +\xi_{4}-\xi_{7}))=0.
  \end{multline*}

A check with Mathematica shows that the co-rank of the total system is three for every admissible values of $s$ and $t$.  The free parameters correspond to two deformation parameters $s$ and $t$ and the uniform scaling of the system.

This approach with the use of Mathematica (see the programme in Appendix B to \cite{Ver}) allows us to prove that the isolated examples of $\vee$-systems 
from the list \cite{FV2} are indeed isolated.

\begin{thm}
\label{(Rigid systems)} There are no non-trivial deformations of the $\vee$-systems
$(E_{7},A_{1}^{2}\times A_{2})$, $(E_{8},A_{2}\times A_{3})$, $(E_{8},A_{2}^{2}\times A_{1})$,
$(E_{8},A_{1}^{3}\times A_{2})$, $(E_{8},A_{1}^{2}\times A_{3})$,
$(E_{8},A_{1}\times A_{4})$, $(H_{4},A_{1})$ and $H_{3}.$
\end{thm}

\section{Matroidal Structure of $\vee$-Systems and Projective Geometry} 

The main part of the classification problem is to characterise the corresponding class of possible matroids. 
This question was addressed by Lechtenfeld et al in \cite{LST}. They
developed a Mathematica program, which generates simple and connected matroids of a given size of the ground set $X$. If a generated matroid has a vector representation, they have checked first if the orthogonality $\vee$-conditions are possible to satisfy before verification of the $\vee$-conditions for the non-trivial planes (all $2$-flats). 
For matroids with $n<10$ elements the orthogonality conditions are strong enough to identify all matroids corresponding to $\vee$-systems in dimensions three. All the identified $\vee$-systems turned out to be part of the list in \cite{FV}. 

For larger matroids this approach seems unworkable because of the unreasonably large computer time required.
This means that we need a more conceptual approach, which is still missing. 

In this section we collect some partial observations 
based on the analysis of the known 3D $\vee$-systems and projective geometry.

We start with the notion of {\it extension and degeneration} for $\vee$-systems.

Let $\mathcal{A_{\mathrm{1}}\mathrm{,}\mathbf{\mathcal{A}_{\mathrm{2}}\subset V^{*}}}$
be two $\vee$-systems. If $\mathcal{A}_{2}  \subset \mathcal{A}_{1}$
we call $\mathcal{A}_{1}$ an {\it extension} of $\mathcal{A}_{2}$.

Let $\vee$-system $\mathcal{A} = \mathcal{A}_{t}$ depend on the parameter $t.$ Assume that for some $t=t_{0}$ one or more of the covectors $\alpha \in \mathcal{A}_{t_{0}}$ vanishes.
In that case the system   $\mathcal{\widetilde{A}} = \underset{t\rightarrow{t_0}}{lim}\mathcal{A}(t)$ is called {\it degeneration} of $\mathcal{A}(t).$
A reverse process we will call {\it regeneration}.

In the tables below we give the list of all extensions and degenerations for known three-dimensional $\vee$-systems from the catalogue in the Appendix.

\begin{table}[H]
	\centering
	\caption{Extensions of known 3D $\vee$-systems.}
\begin{tabular}{|c|c|c|}
\hline 
$\vee$-system  & Extension & The added covectors\tabularnewline
\hline 
\hline 
 $A_{3}$ & $F_{3}(t)$  & $\{1,2,3,10,11,12,13\}$\tabularnewline
\hline 
$A_{3}$ & $(AB_{4}(t),A_{1})_{2}$  & $\{1,2,3,10\}$\tabularnewline
\hline 
$(E_{6},A_{1}^{3})$ & $(E_{8},A_{1}^{3}\times A_{2})$  & $\{3,4,5,6,9,14,15,16,17\}$\tabularnewline
\hline 
 $G_{3}(\frac{3}{2})$ & $(E_{8},A_{1}^{3}\times A_{2})$ &$\{1,2,10,11,12,13\}$\tabularnewline
\hline 
 $H_{3}$ &$(H_{4},A_{1})$& $\{4,5,6,7,20,21,22,23,24,25,26,27,28,29,30\}$\tabularnewline
\hline 
\end{tabular}
\end{table}

\begin{table}[H]
	\centering
	\caption{Degenerations of known 3D $\vee$-systems.}
\begin{tabular}{|c|c|c|}
\hline 
$\vee$-system & Degeneration & The vanishing covectors\tabularnewline
\hline 
\hline 
$F_{3}(t)$ & $\underset{t\rightarrow 0}{lim}\, F_{3}(t)\sim B_{3}(\sqrt{2})$ & $\{11,12,13,14\}$\tabularnewline
\hline 
$F_{3}(t)$ & $\underset{t\rightarrow \infty}{lim} F_{3}(t) \sim D_3(1,1)$ & $\{8,9,10,11,12,13\}$\tabularnewline
\hline 
$B_{3}(c,c,c;\gamma)$ & $\underset{\gamma\rightarrow c}{lim} \,B_{3}(c,c,c;\gamma)\sim A_{3}$ & $\{1,2,3\}$\tabularnewline
\hline 
$(AB_{4}(t),A_{1})_{2}$ & $\underset{t\rightarrow\infty}{lim}\,(AB_{4}(t),A_{1})_{2}\sim D_3(1,1)$ & $\{5,7,9\}$\tabularnewline
\hline 
$(AB_{4}(t),A_{1})_{2}$ & $\underset{t\rightarrow0}{lim}\,(AB_{4}(t),A_{1})_{2}\sim B_{3}(\sqrt{2})$ & $\{10\}$\tabularnewline
\hline 
$(AB_{4}(t),A_{1})_{1}$ & $\underset{t\rightarrow\frac{1}{\sqrt{2}}}{lim}\,(AB_{4}(t),A_{1})_{1}\sim(E_{6},A_{1}^{3})$ & $\{3\}$\tabularnewline
\hline 
$\frac{1}{t}(AB_{4}(t),A_{1})_{1}$ & $\underset{t\rightarrow\infty}{lim}\,(AB_{4}(t),A_{1})_{1}\sim B_{3}(\sqrt{2})$ & $\{5,6\}$\tabularnewline
\hline 
$G_{3}(t)$ & $\underset{t\rightarrow\frac{1}{2}}{lim}\,G_{3}(t)\sim(E_{6},A_{1}^{3})$ & $\{4,5,6\}$\tabularnewline
\hline 
$D_{3}(t,s)$ & $\underset{t\rightarrow\ (s+1)}{lim}\,D_{3}(t,s)\sim A_{3}$ & $\{5\}$\tabularnewline
\hline 
$B_{3}(c_1,c_2,c_3,\gamma)$ & $\underset{c_1\rightarrow\ -\gamma}{lim}\,B_{3}(c_1,c_2,c_3,\gamma)\sim (E_6,A_3)$ & $\{1\}$\tabularnewline
\hline 
$B_{3}(c_1,c_2,c_3,\gamma)$ & $\underset{c_1,c_2\rightarrow\ -\gamma}{lim}\,B_{3}(c_1,c_2,c_3,\gamma)\sim D_3(1,1) $ & $\{1,2\}$\tabularnewline
\hline
\end{tabular}
\end{table}

More relations between $\vee$-realisable matroids can be seen using projective geometry, in particular projective duality.
We will demonstrate this on few examples.

We start with the matroid of the $\vee$-system of type $A_{3}.$ In projective geometry (see e.g. \cite{HCV}) it is known as the simplest configuration  $(6_{2}4_{3})$ consisting of four lines with three points on each line and two lines passing through every point. Its projective dual is a complete quadrangle $(4_{3}6_{2})$ consisting of four points, no three of which are collinear and six lines connecting each pair of points (see figure \ref{D1}). If we extend the dual configuration  by adding the remaining three points of intersections of lines (the points marked white in the graphic), we come to the projective configuration of seven points and six lines, corresponding to the matroid of the $\vee$-system of type $D_3$. 

\begin{figure} [H]
\begin{center}
\includegraphics[scale=0.25]{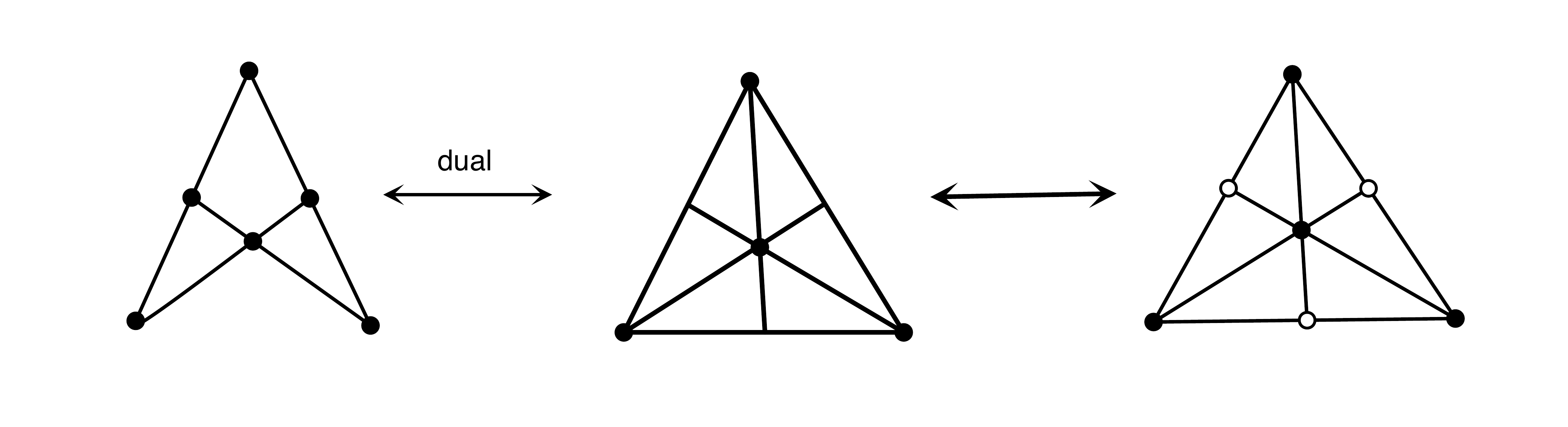}
\caption{The projective configuration of $A_3$ type, its dual and the extended configuration corresponding to $\vee$-system of type $D_3$.} 
\label{D1}
\end{center}
\end{figure}

We can proceed the construction by taking the dual of the new obtained configuration and extending it by adding the missing points of intersections of lines. The result is the configuration of nine points and seven lines realisable as $B_{3}$-type $\vee$-system (see figure \ref{D2}).

\begin{figure} [H]
\begin{center}
\includegraphics[scale=0.23]{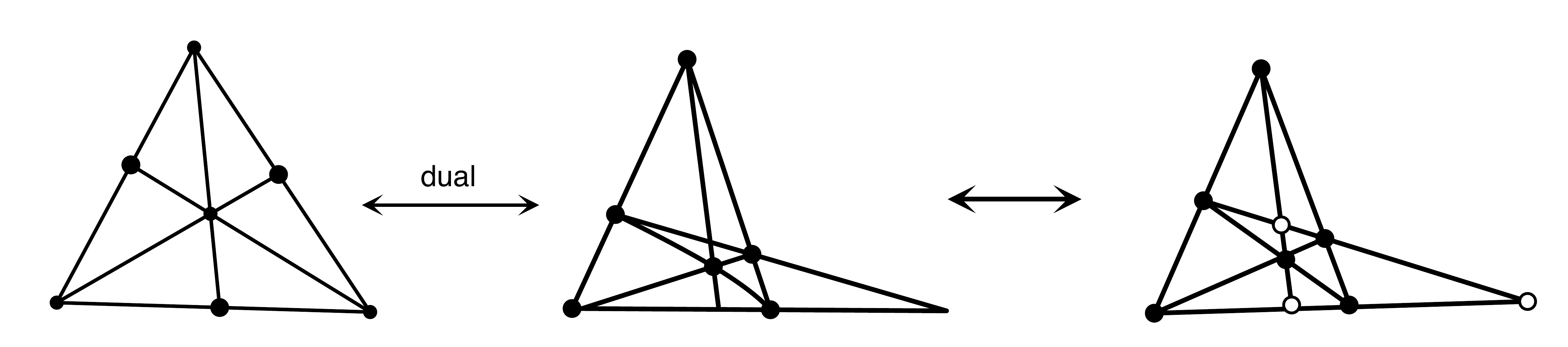}
\caption{$D_{3}$ configuration, its projective dual and the extended configuration of matroidal type $B_3.$} 
\label{D2}
\end{center}
\end{figure}

The next step of the construction is demonstrated in figure \ref{D3}. The dual configuration was obtained from the configuration $D_{3}$ by adding all missing lines passing through any pair of points. Applying Desargue's theorem to two marked triangles we see that the white marked points of the extended configuration are collinear. The new configuration of $10$ points and $10$ lines is self-dual and corresponds to the $\vee$-system of type $(AB_{4} (t),A_{1})_{2}$ (see system 9.6 in the Appendix).

\begin{figure} [H]
\begin{center}
\centerline{ \includegraphics[scale=0.23]{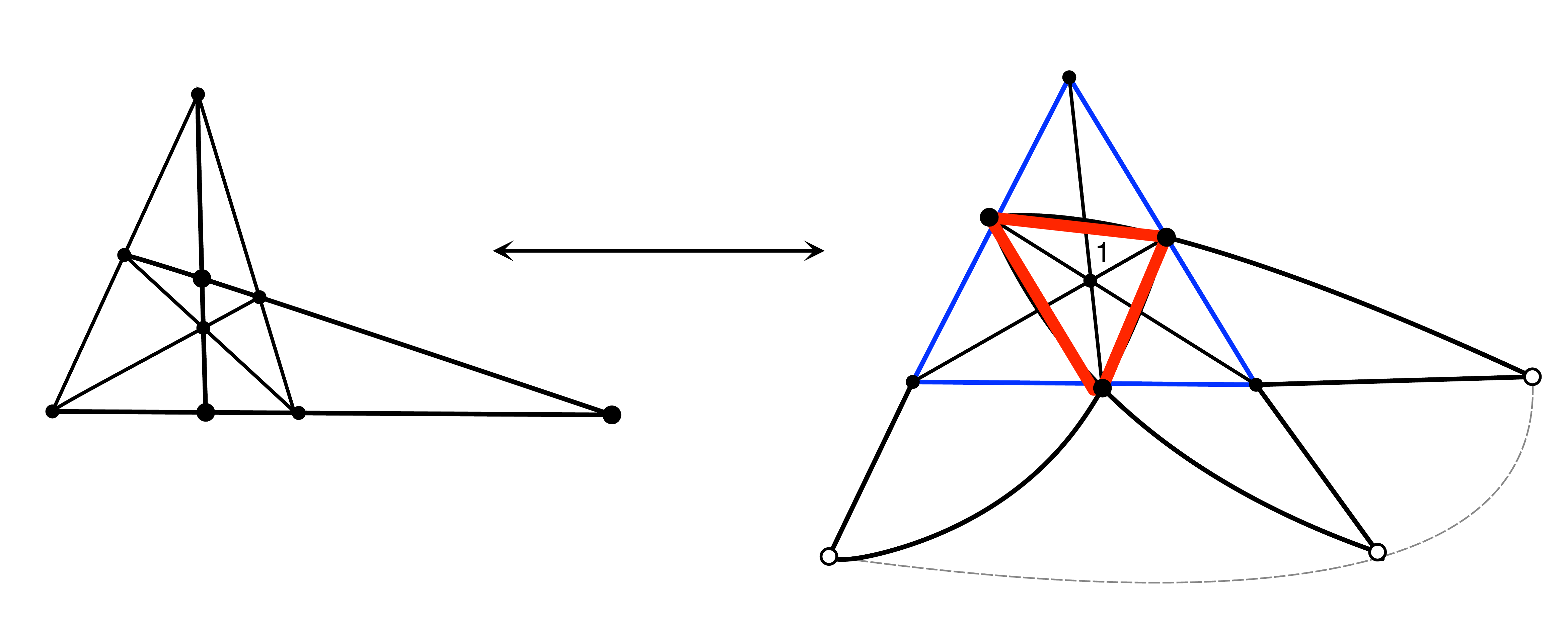}}
\caption{$B_{3}$-configuration and its schematic extended projective dual, corresponding to $\vee$-system of type $(AB_{4} (t),A_{1})_{2}$. Due to Desargue's theorem the three added (white) points are collinear.} 
\label{D3}
\end{center}
\end{figure}

One can check that the adding of three intersection points with red lines and three lines connecting them pairwise leads to the configuration of $F_3$ type (see system 9.10). However, if we add also the three intersection points with dotted line then we come to the configuration which can be shown to be not $\vee$-realisable.

Although this relation with projective configurations and theorems in projective geometry looks quite promising, we see that the extension procedure is not straightforward and does not guarantee the $\vee$-realisability of the resulting configuration. 

We conclude with the following conjecture about 2-flats with precisely four points.
Recall that four points $A,B,C,D$ on a projective line form a {\it harmonic range} if the cross-ratio
$(A,B;C,D)=-1.$ The corresponding pencil of four lines on a plane is called {\it harmonic bundle.}
The $B_3$ configuration provides a geometric way to construct harmonic ranges: 
on Fig. 1 the points 3,4,9,8 always form a harmonic range.
Note that the covectors 8 and 9 are orthogonal and determine 
the bisectors for the lines corresponding to covectors 3 and 4.
Case by case check of the known 3D $\vee$-systems suggests 
that the same is true in general.

\begin{conj}
Let $\mathcal{A}$ be a $\vee$-system and $\Pi_{\mathcal{A}}\subset V^{*}$
be two-dimensional plane containing exactly four covectors $\alpha_{i}\in\mathcal{A}$,
$i=1,\ldots,4$, then the corresponding four lines form a harmonic
bundle with two orthogonals. 
\end{conj}

\section{$\nu$-Function, Uniqueness and Rigidity Conjectures}

Let $M$ be a matroid and $\mathcal A$ be its $\vee$-system realisation.
Such a realisation defines the {\it $\nu$-function} on the 2-flats of $M$, 
where $\nu$ is the coefficient in the $\vee$-conditions (\ref{vee}) 
corresponding to the plane $\Pi$ representing the flat.

\begin{conj}
\label{(Uniqueness conj)} {\bf (Uniqueness Conjecture)}
An irreducible $\vee$-system $\mathcal A$ is uniquely determined modulo linear group $GL(V^*)$ by its matroid $M$ 
and the corresponding $\nu$-function on its flats.
\end{conj}

A weaker version of the conjecture is

\begin{conj}
\label{(Rigidity conj)} {\bf (Rigidity Conjecture)}
An irreducible $\vee$-system $\mathcal A$ is locally uniquely determined by its matroid $M$ 
and the corresponding $\nu$-function on its flats.
\end{conj}

If the function $\nu$ is fixed under deformation then $\dot \nu=0$ and the corresponding 
$\vee$-conditions are 
\begin{equation}
\label{1XX}
X(\alpha^{\vee}, \beta^{\vee})=0
\end{equation}
for $\alpha,\beta$ be the only two covectors in the plane, and
\begin{equation}
\label{2XX}
X^{\Pi}-\nu X=0\mid_{\Pi^{\vee}}
\end{equation}
for any plane $\Pi$ containing more than two covectors from $\mathcal A.$

Conjecturally this should imply that $X=cG$ corresponding 
to the global scaling of the system.

Case by case check from the list in the Appendix leads to the following 

\begin{thm}
\label{conj}
Both conjectures are true for all known $\vee$-systems in dimension three.
\end{thm}

We have also the following conjecture based on the analysis of the list of extensions of $\vee$-systems from the previous section.

\begin{conj}
\label{(Extension conj)} {\bf (Extension Conjecture)}
For any irreducible $\vee$-system and its extension the values of the $\nu$-functions on the corresponding flats are proportional.
\end{conj}

One can check that this is true for all known cases. For example, for the extension $H_3 \subset (H_{4},A_{1})$ we have the set of values
$$ \{3/10, 1/2\}=3 \times \{1/10 , 1/6\}.$$

Now we present some results about $\nu$-functions for $\vee$-systems.

First we give the following, more direct geometric way to compute $\nu(\Pi).$
The form $G_{\mathcal A}$ on $V$  
defines the scalar product on $V^*$ and thus the norm $|\alpha|, \, \alpha \in V^*.$

\begin{thm}
\label{geome}
For every plane $\Pi \subset V^*$ containing more than two covectors $\alpha$ from a $\vee$-system $\mathcal A$
\begin{equation}
\nu (\Pi)=\frac{1}{2}\sum_{\alpha\in\Pi\cap \mathcal A}|\alpha|^2. 
\label{UC2}
\end{equation}
\end{thm}

\begin{proof}
From the $\vee$-conditions (1.1) we have
\[
\sum_{\alpha\in\Pi\cap \mathcal A}\alpha^{\vee}\otimes\alpha \mid_{\Pi^{\vee}}=\nu(\Pi) I \mid_{\Pi^{\vee}}.
\]
Taking the trace of both sides gives (\ref{UC2}).
\end{proof}

Let $\mathcal{A}\subset V^{*}$ be a $\vee$-system generating $V^*$ and consider the set $\mathcal F_{\mathcal A}$ of 2-flats in the corresponding matroid, which the same as the set of 2D planes $\Pi \subset V^*$ containing more than two covectors from $\mathcal A$.

We say that the set of weights $x_{\Pi}, \, \Pi \in \mathcal F_{\mathcal A}$ is {\it admissible} if for each $\alpha\in\mathcal{A}$ 
\begin{equation}
\sum_{\Pi \in \mathcal F_{\mathcal A}: \alpha \in \Pi}x_{\Pi}=1.\label{eq:Summ_x}
\end{equation}

\begin{thm}
\label{thm:Universal Relation 1} For every admissible set of weights we have
\begin{equation}
\sum_{\Pi \in \mathcal F_{\mathcal A}}x_{\Pi}\nu(\Pi)=\frac{n}{2},
\label{UR1}
\end{equation}
where $n$ is the dimension of $V.$
\end{thm}

\begin{proof}
We have 
\[
\sum_{\beta\in\mathcal{A}}\beta^{\vee}\otimes\beta=\sum_{\beta\in\mathcal{A}}(\sum_{\Pi \in \mathcal F_{\mathcal A}: \beta \in \Pi}x_{\Pi})\beta^{\vee}\otimes\beta=\sum_{\Pi\in \mathcal F_{\mathcal A}}x_{\Pi}\sum_{\beta\in\Pi\cap\mathcal{A}}\beta^{\vee}\otimes\beta.
\]
From the $\vee$-condition 
$$\sum_{\beta\in\Pi\cap\mathcal{A}}\beta^{\vee}\otimes\beta=\nu(\Pi)P_{\Pi},$$
where $P_{\Pi}$ is the orthogonal projector onto $\Pi^{\vee}.$
Taking trace and using the fact that
$
\sum_{\beta\in\mathcal{A}}\beta^{\vee}\otimes\beta=Id
$
we obtain (\ref{UR1}).\end{proof}

We call (\ref{UR1}) the {\it universal relation} for values of function $\nu.$

For the $\vee$-system of type $A_3$ the universal relation completely describes the set of all possible functions $\nu.$
Indeed, one can easily see from Fig. 2  that $x_{\Pi}=1/2$ is the only admissible weight system, which leads to the universal relation
$$\sum_{\Pi \in \mathcal F_{\mathcal A}}\nu(\Pi)=3.$$
This gives us three free parameters, which are exactly three parameters of deformation.

However, in general universal relations are not strong enough to describe possible $\nu$-functions. 
Moreover, a $\vee$-system $\mathcal{A}$
may not have admissible weights $x_{\Pi}$ at all. 
For instance, this is the case for the $\vee$-system of $D_{3}(t)$-type
(this is however the only exception among known 3D $\vee$-systems).

The list of all known $\vee$-systems in dimension three with the corresponding $\nu$-functions is given in the Appendix.

\section{Concluding Remarks}

Although the problem of classification of $\vee$-systems seems to be very hard, in
dimension three it does not look hopeless. As we have seen, matroid theory provides a
natural framework for the problem of classification of $\vee$-systems. For a given
matroidal structure the $\vee$-conditions define a set of algebraic relations 
on the vector realisations. In case when matroid is strongly projectively rigid 
we have one free parameter for each vector, which makes possible 
the full classification of $\vee$-systems with small number of vectors. 

The main problem is to describe all possible matroidal types, 
which we believe form a finite list in any dimension.
The results of Lechtenfeld et al \cite{LST} show that 
the direct computer approach is probably unrealistic for $\vee$-systems 
with more than 10 covectors, while we have already in dimension three 
an example with 31 covectors (system $(H_4,A_1)$, see 7.16 in the Appendix).
In dimension three we have an intriguing relation with the theory of configurations
on the projective plane and with the theorems in projective geometry, 
which also suggests that the final list should be finite.

In the theory of matroids and graphs many families have been proved 
to be closed under taking minors, thus giving a possibility to reduce the problem 
of classification to the identification of the forbidden minors. We hope that a similar approach could be fruitful for classification of $\vee$-systems.

Another result from matroid theory, which could be relevant, is Seymour's decomposition
theorem \cite{Sey}, which states that all regular matroids can be build up
in a simple way as sums of certain type of graphic
matroids, their duals, and one special matroid on 10 elements. Our analysis
of degenerations and extensions of $\vee$-systems suggests a possibility
of a similar result for the $\vee$-realisable matroids.

\section{Acknowledgements}
We are grateful to Olaf Lechtenfeld and especially to Misha Feigin for useful discussions.
We thank also an anonymous referee for a very thorough job.

The work of APV was partly supported by the EPSRC (grant EP/J00488X/1).

\section{Appendix. Catalogue of all Known Real $3$-Dimensional $\vee$-Systems}

Each 3D $\vee$-system $\mathcal A$ is presented below by the matrix with columns giving the covectors of the system (the first row is simply the labelling of the covectors). We give the graphical representation of the corresponding matroid with the list of orthogonal pairs, 2-flats, the form $G$ and the values of $\nu$-function. The ordering of the list is according to the number of covectors in the system. The parameters are assumed to be chosen in such a way that all the covectors are real and non-zero.

Below is a schematic way to present all known $\vee$-systems in dimension three taken from \cite{FV2}.

\begin{figure} [H]
\begin{center}
\includegraphics[scale=0.7]{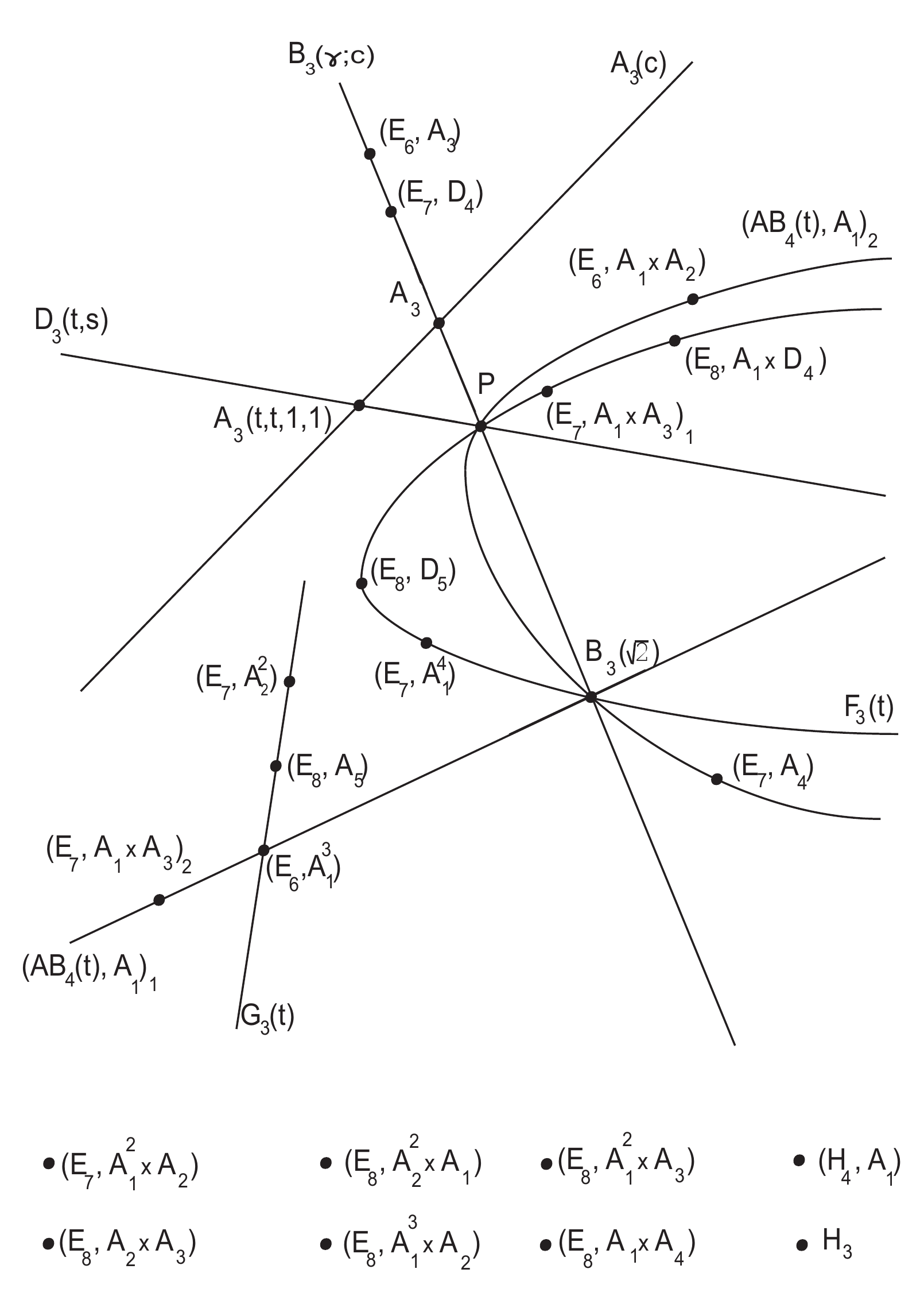} 
\caption{The map of all known 3-dimensional $\vee$-systems from \cite{FV2}.}
\end{center}
\end{figure}

We use here the notations from \cite{FV, FV2}. In particular, for a Coxeter group  $G$ and its parabolic subgroup $H$ $(G,H)$ denotes the corresponding $\vee$-system given by the restriction procedure \cite{FV}. When the type of the subgroup does not fix the subgroup up to a conjugation the index 1 or 2 is used to distinguish them. 

The $\vee$-systems of type $AB_4$, $G_3$ and $D_3$ are related to the exceptional generalised root systems $AB(1,3)$, $G(1,2)$ and $D(2,1,\lambda)$ appeared in the theory of basic classical Lie superalgebras \cite{Sergan, SV}.

\subsection{$\vee$-systems $A_{3}(c_1,c_2,c_3)$}

\begin{eqnarray*}
\mathcal{A}
 & =
 & \left[\begin{array}{cccccc}
1 & 2 & 3 & 4 & 5 & 6\\
\sqrt{c_{1}} & 0 & 0 & -\sqrt{c_{1}c_{2}} & -\sqrt{c_{1}c_{3}} & 0\\
0 & \sqrt{c_{2}} & 0 & \sqrt{c_{1}c_{2}} & 0 & -\sqrt{c_{2}c_{3}}\\
0 & 0 & \sqrt{c_{3}} & 0 & \sqrt{c_{1}c_{3}} & \sqrt{c_{2}c_{3}}
\end{array}\right]
\end{eqnarray*}

\begin{figure} [H]
\begin{center}
\includegraphics[scale=0.65]{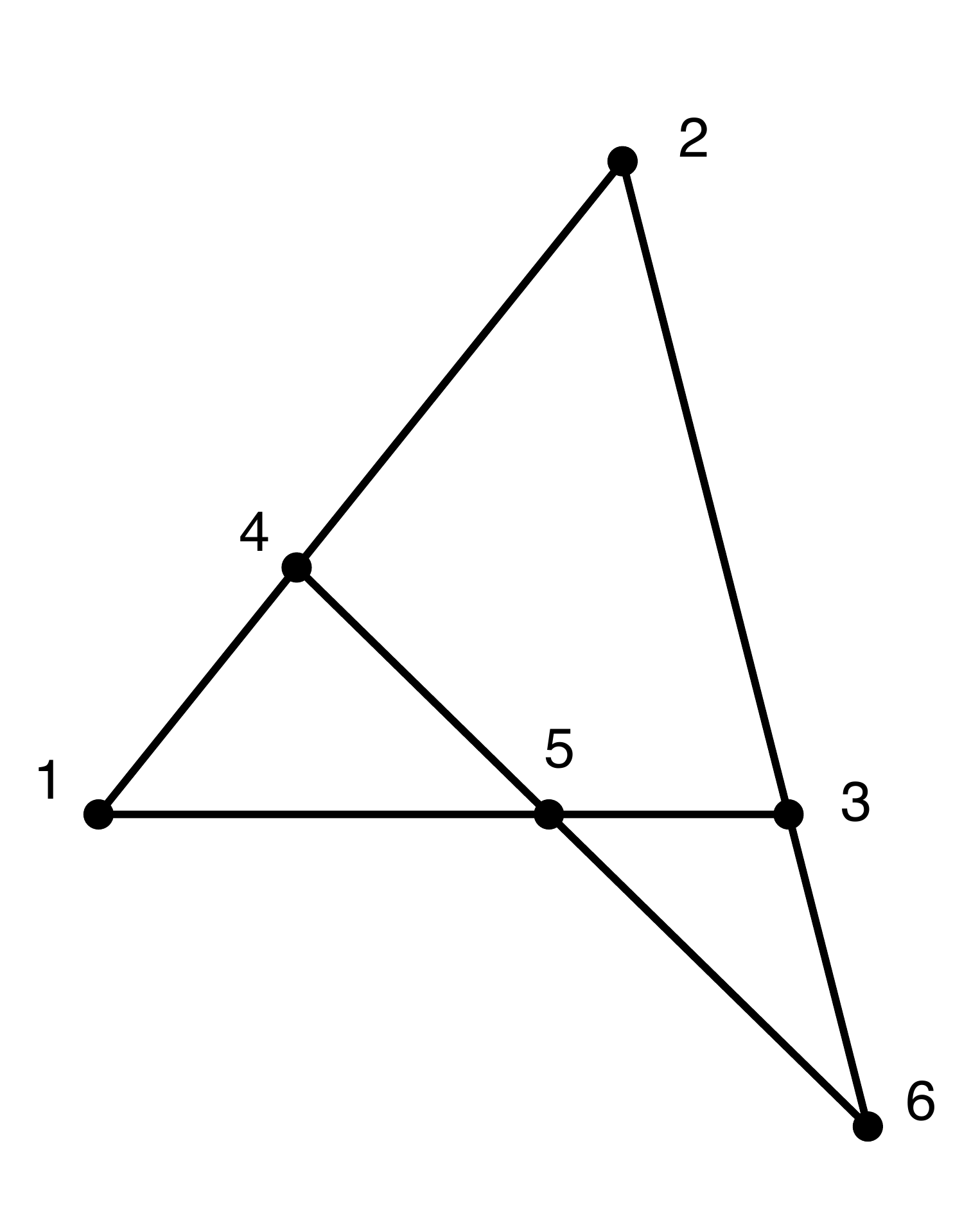}
\end{center}
\end{figure}
 
$$G = \left[\begin{array}{ccc}
c_{1}(1+c_{2}+c_{3}) & -c_{1}c_{2} & -c_{1}c_{3}\\
-c_{1}c_{2} & c_{2}(1+c_{1}+c_{3}) & -c_{2}c_{3}\\
-c_{1}c_{3} & -c_{2}c_{3} & c_{3}(1+c_{1}+c_{2})
\end{array}\right]$$\medskip

$\mathcal{I_{\mathrm{2}}}$ =$\left\{ (1,6),(2,5),(3,4)\right\} $\medskip

$\mathcal{I_{\mathrm{3}}}$ = $\begin{cases}
(1,2,4) & \nu_{3}=\frac{1+c_1+c_2}{1+c}\\
(1,3,5) & \nu_{2}=\frac{1+c_1+c_3}{1+c}\\
(2,3,6) & \nu_{1}=\frac{1+c_2+c_3}{1+c}\\
(4,5,6), & \nu_{4}=\frac{c}{1+c}, \, c=c_1+c_2+c_3.
\end{cases}$
\bigskip{}

\subsection{$\vee$-system $D_{3}(t,s)$}

\bigskip

$$\mathcal{A} =\left[\begin{array}{ccccccc}
1 & 2 & 3 & 4 & 5 & 6 & 7\\
1 & 1 & 1 & 1 & 0 & 0 & \sqrt{2(s+t-1)}\\
1 & -1 & -1 & 1 & \sqrt{\frac{2(s-t+1)}{t}} & 0 & 0\\
1 & -1 & 1 & -1 & 0 & \sqrt{\frac{2(t-s+1)}{s}} & 0
\end{array}\right]$$\medskip{}

\begin{figure} [H]
\begin{center}
\includegraphics[scale=0.6]{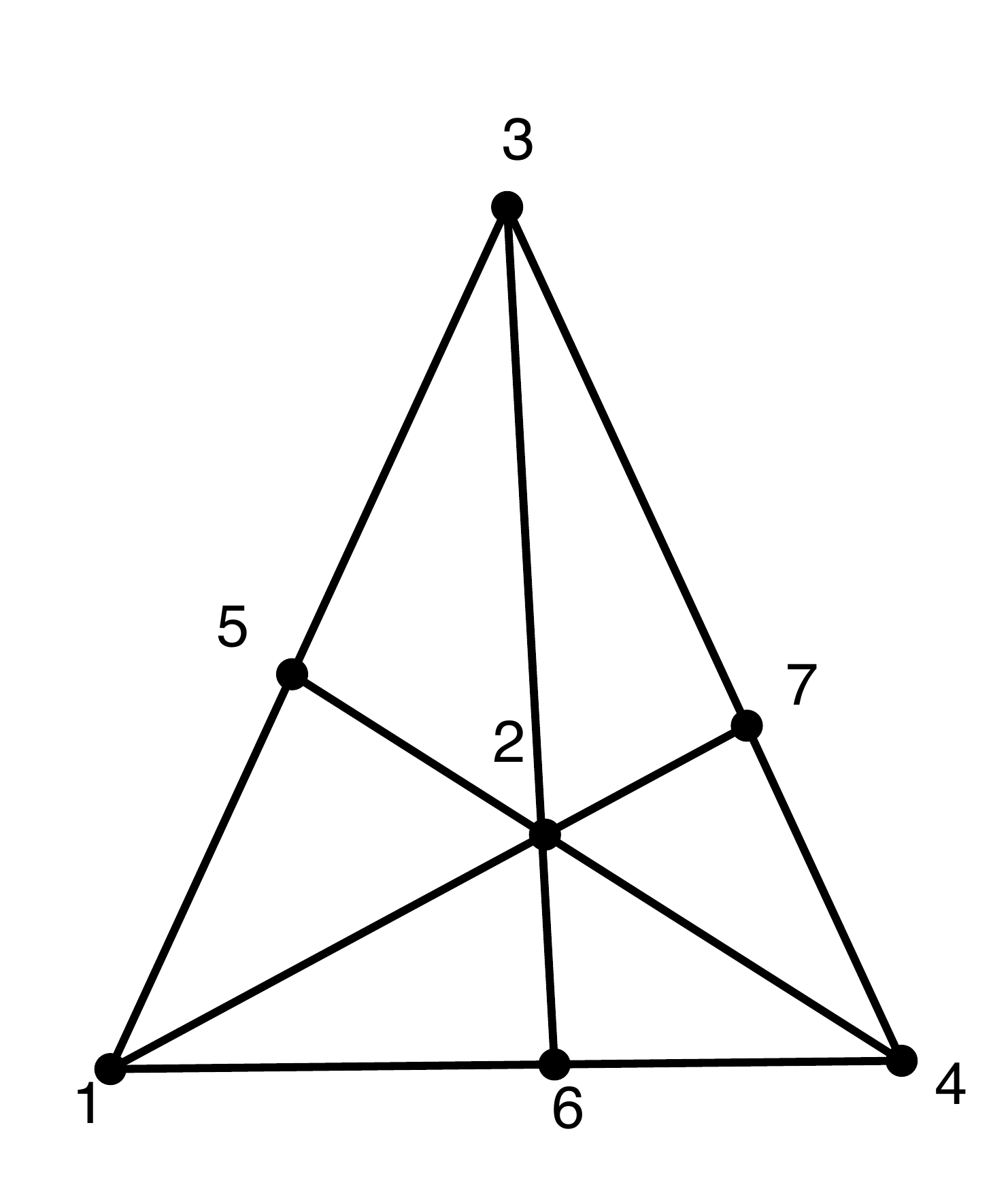}
\end{center}
\end{figure}

\medskip{}

$$G = 2(1+s+t)\left[\begin{array}{ccc}
1 & 0 & 0\\
0 & \frac{1}{t} & 0\\
0 & 0 & \frac{1}{s}
\end{array}\right]$$
\medskip{}

$\mathcal{I_{\mathrm{2}}}$ =$\left\{ (5,6),(5,7),(6,7)\right\} $

\medskip{}

$\mathcal{I}_{3}$ = $\begin{cases}
(1,2,7),(3,4,7) & \nu_{31}=\frac{s+t}{1+s+t}\\
(1,3,5),(2,4,5) & \nu_{32}=\frac{1+s}{1+s+t}\\
(1,4,6),(2,3,6) & \nu_{33}=\frac{1+t}{1+s+t}
\end{cases}$\medskip{}

\subsection{$\vee$-system $(E_{6},A_{3})$}

\medskip{}

\[
\mathcal{A}=\left[\begin{array}{cccccccc}
1 & 2 & 3 & 4 & 5 & 6 & 7 & 8\\
\\
2 & 2 & 2\sqrt{6} & 2 & -2 & 2 & -2 & 0\\
2 & -2 & 0 & \frac{1}{2} & \frac{1}{2} & -1 & -1 & \frac{\sqrt{6}}{2}\\
0 & 0 & 0 & \frac{1}{2} & \frac{1}{2} & 1 & 1 & \frac{\sqrt{6}}{2}
\end{array}\right]
\]
\medskip{}

\begin{figure} [H]
\begin{center}
\includegraphics[scale=0.6]{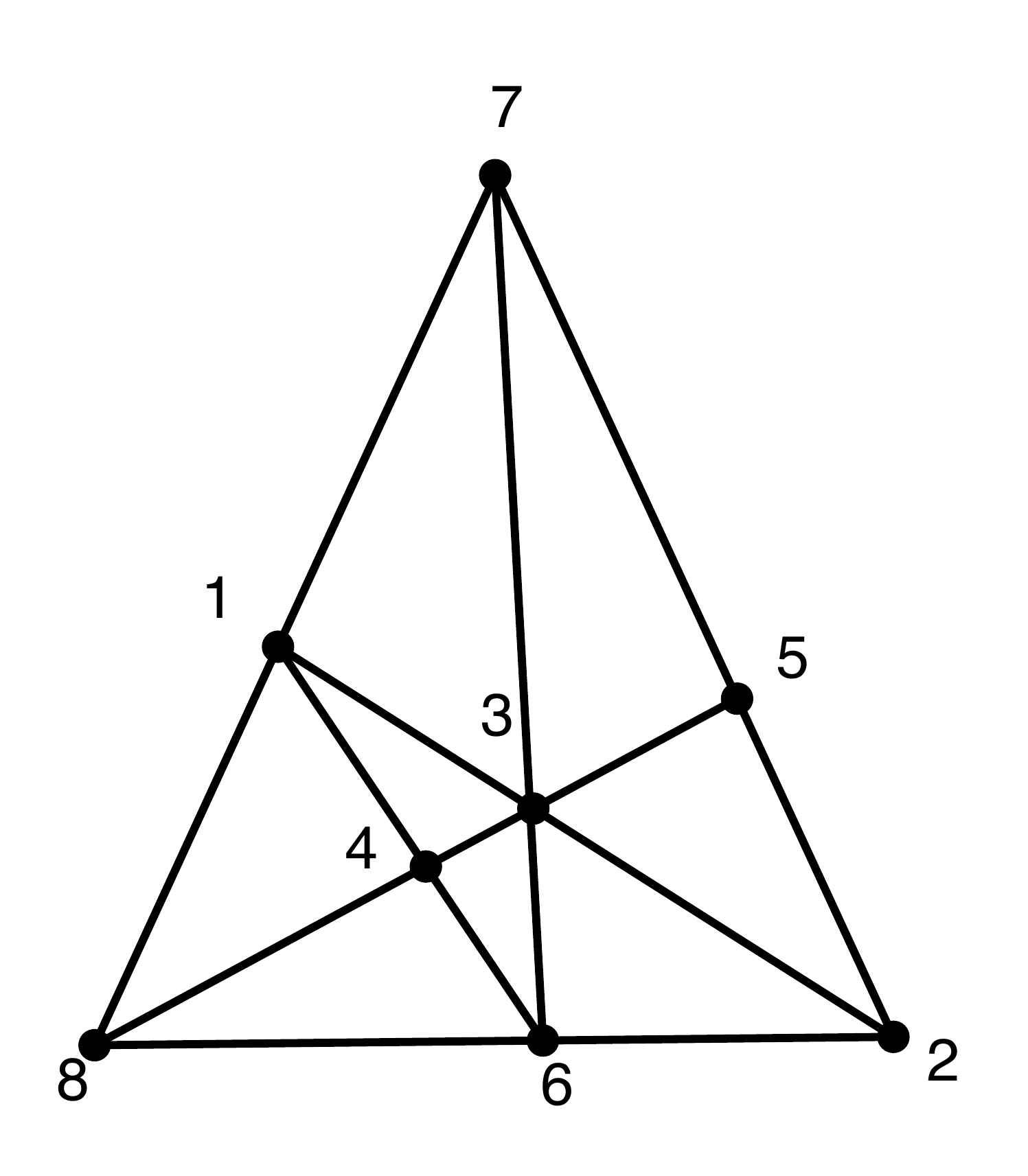}
\end{center}
\end{figure}

\medskip{}

G = $4\left[\begin{array}{ccc}
12 & 0 & 0\\
0 & 3 & 0\\
0 & 0 & 1
\end{array}\right]$\medskip{}

$\mathcal{I_{\mathrm{2}}}$ =$\left\{ (1,5),(2,4),(4,7),(5,6)\right\} $\medskip{}

$\mathcal{I_{\mathrm{3}}}$ =$\begin{cases}
(2,5,7),(1,4,6) & \nu_{31}=\frac{1}{2}\\
(6,3,7),(8,1,7),(1,2,3),(8,6,2) & \nu_{32}=\frac{2}{3}
\end{cases}$\medskip{}

$\mathcal{I_{\mathrm{4}}}$ =$\begin{cases}
(8,4,3,5) & \nu_{4}=\frac{2}{3}\end{cases}$\medskip{}

\subsection{$\vee$-systems of $B_{3}(c_1,c_2,c_3, \gamma)$}

\begin{eqnarray*}
\mathcal{A}
 & =
 & \left[\begin{array}{ccc}
1 & 2 & 3\\
\\
\sqrt{2c_{1}(c_{1}+\gamma)} & 0 & 0\\
0 & \sqrt{2c_{2}(c_{2}+\gamma)} & 0\\
0 & 0 & \sqrt{2c_{3}(c_{3}+\gamma)}
\end{array}\cdots\right.\\
 & \begin{array}{cccccc}
\end{array} & \left.\cdots\begin{array}{c}
\begin{array}{cccccc}
4 & 5 & 6 & 7 & 8 & 9\\
\\
0 & \sqrt{c_{1}c_{3}} & \sqrt{c_{1}c_{2}} & 0 & -\sqrt{c_{1}c_{3}} & -\sqrt{c_{1}c_{2}}\\
\sqrt{c_{2}c_{3}} & 0 & \sqrt{c_{1}c_{2}} & -\sqrt{c_{2}c_{3}} & 0 & \sqrt{c_{1}c_{2}}\\
\sqrt{c_{2}c_{3}} & \sqrt{c_{1}c_{3}} & 0 & \sqrt{c_{2}c_{3}} & \sqrt{c_{1}c_{3}} & 0
\end{array}\end{array}\right]
\end{eqnarray*}
\begin{figure} [H]
\begin{center}
\includegraphics[scale=0.55]{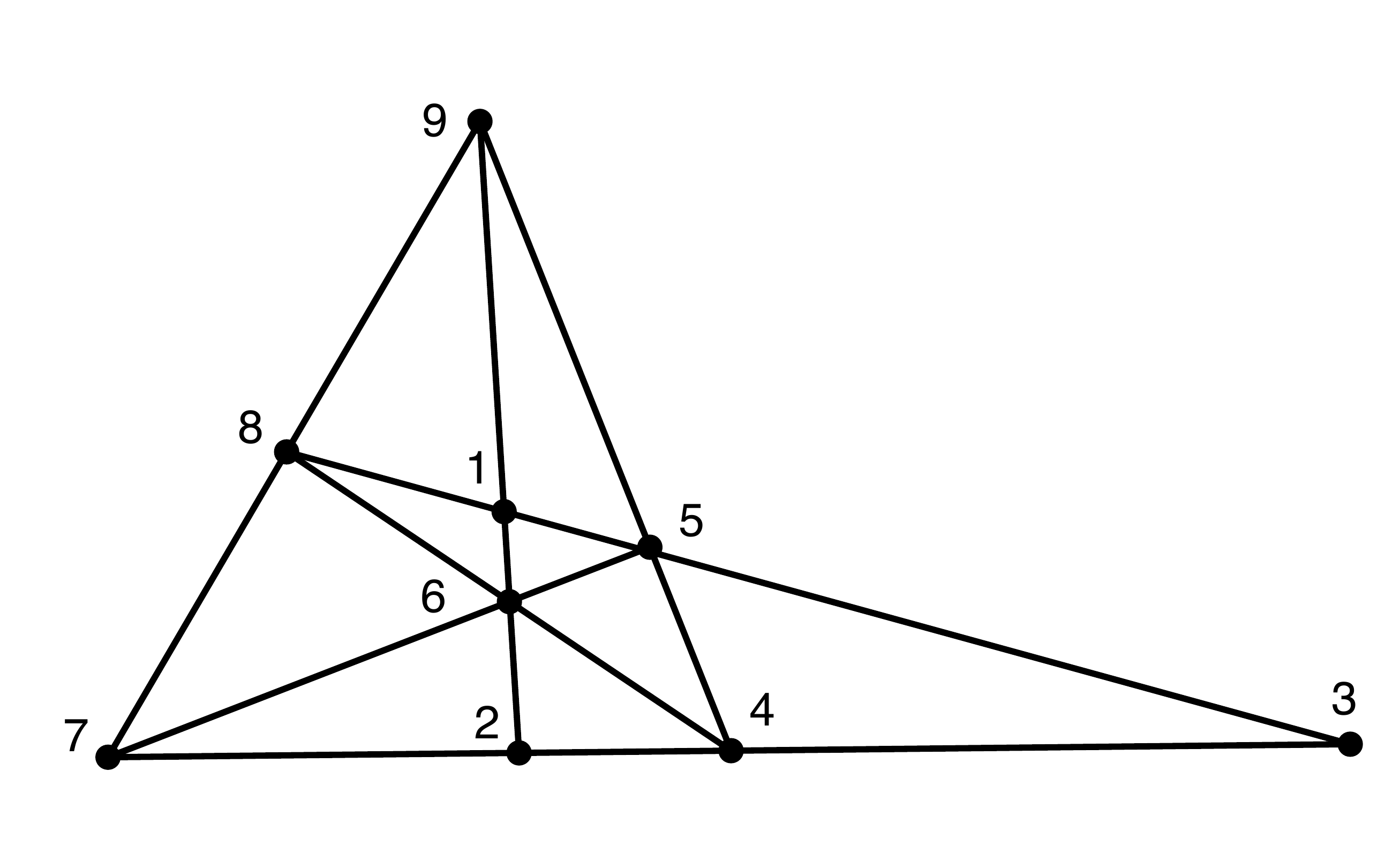}
\end{center}
\end{figure}
$$G = 2(c_{1}+c_{2}+c_{3}+\gamma)\left[\begin{array}{ccc}
c_{1} & 0 & 0\\
0 & c_{2} & 0\\
0 & 0 & c_{3}
\end{array}\right]$$

$\mathcal{I_{\mathrm{2}}}$ =$\left\{ (1,4),(1,7),(2,5),(2,8),(3,6),(3,9)\right\} $

Four 3-point lines:

$\mathcal{I_{\mathrm{3}}}$ =$\left\{ (4,5,9),(4,6,8),(5,6,7),(7,8,9)\right\},$
$\nu_{3}=\frac{\sum c_{i}}{2(\gamma+\sum_{i=1}^{3}c_{i})}$

Three 4-point lines:

$\mathcal{I_{\mathrm{4}}}$ =$\left\{ (1,2,6,9),(1,3,5,8),(2,3,4,7)\right\}, $
$\nu_{4j}=\frac{\gamma-c_{j}+\sum c_{i}}{\gamma+\sum_{i=1}^{3}c_{i})}$, \,j=1,2,3\medskip{}

\subsection{$\vee$-system $(E_{6},A_{1}^{3})$}

\medskip{}
\begin{eqnarray*}
\mathcal{A}
 & =
 & \left[\begin{array}{cccccccccc}
1 & 2 & 3 & 4 & 5 & 6 & 7 & 8 & 9 & 10\\
\\
\sqrt{2} & \sqrt{2} & 2\sqrt{3} & 0 & \sqrt{2} & -\sqrt{2} & -\sqrt{2} & \sqrt{2} & 0 & 0\\
\sqrt{2} & -\sqrt{2} & 0 & 2 & \frac{\sqrt{2}}{2} & -\frac{\sqrt{2}}{2} & \frac{\sqrt{2}}{2} & -\frac{\sqrt{2}}{2} & 1 & -1\\
0 & 0 & 0 & 0 & \frac{\sqrt{2}}{2} & \frac{\sqrt{2}}{2} & \frac{\sqrt{2}}{2} & \frac{\sqrt{2}}{2} & 1 & 1
\end{array}\right]
\end{eqnarray*}
\medskip{}

\begin{figure} [H]
\begin{center}
\includegraphics[scale=0.5]{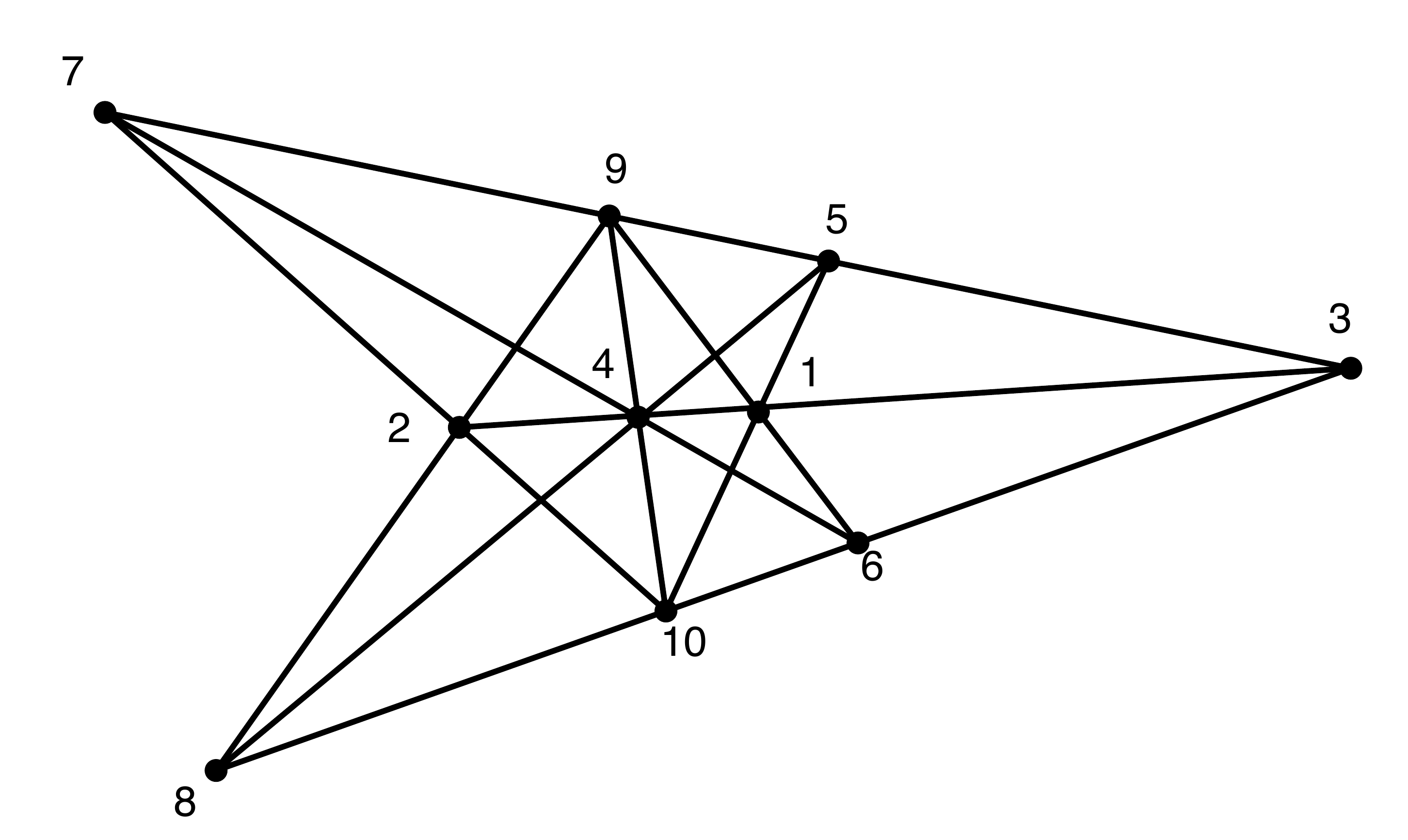}
\end{center}
\end{figure}
\medskip{}

$$G = \left[\begin{array}{ccc}
24 & 0 & 0\\
0 & 12 & 0\\
0 & 0 & 4
\end{array}\right]$$

$\mathcal{I_{\mathrm{2}}}$ =$\left\{ (1,7),(1,8),(2,5),(2,6),(5,6),(7,8)\right\} $

\medskip{}

$\mathcal{I}$$_{3}$ =$\begin{cases}
(10,2,7),(4,6,7),(9,8,2),(1,5,10),(4,5,8),(9,1,6), & \nu_{31}=\frac{5}{12}\\
(4,9,10), & \nu_{32}=\frac{1}{2}
\end{cases}$

\medskip{}

$\mathcal{I}_{4}$ =$\left\{ (4,1,2,3),(9,5,3,7),(6,3,8,10)\right\} ,\nu_{4}=\frac{2}{3}$

\medskip{}

\subsection{$\vee$-system $(AB_{4}(t),A_{1})_{2}$}

\medskip{}

$$\mathcal{A} =\left[\begin{array}{cccccccccc}
1 & 2 & 3 & 4 & 5 & 6 & 7 & 8 & 9 & 10\\
\\
\sqrt{2} & 0 & 0 & 1 & \frac{1}{\sqrt{4t^{2}+1}} & 1 & \frac{1}{\sqrt{4t^{2}+1}} & 0 & 0 & \frac{t\sqrt{2}}{\sqrt{(t^{2}+1)}}\\
0 & \sqrt{2} & 0 & 1 & -\frac{1}{\sqrt{4t^{2}+1}} & 0 & 0 & 1 & \frac{1}{\sqrt{4t^{2}+1}} & \frac{t\sqrt{2}}{\sqrt{(t^{2}+1)}}\\
0 & 0 & \sqrt{2} & 0 & 0 & 1 & -\frac{1}{\sqrt{4t^{2}+1}} & 1 & -\frac{1}{\sqrt{4t^{2}+1}} & \frac{t\sqrt{2}}{\sqrt{(t^{2}+1)}}
\end{array}\right]$$ 

\medskip{}

\begin{figure} [H]
\begin{center}
\includegraphics[scale=0.65]{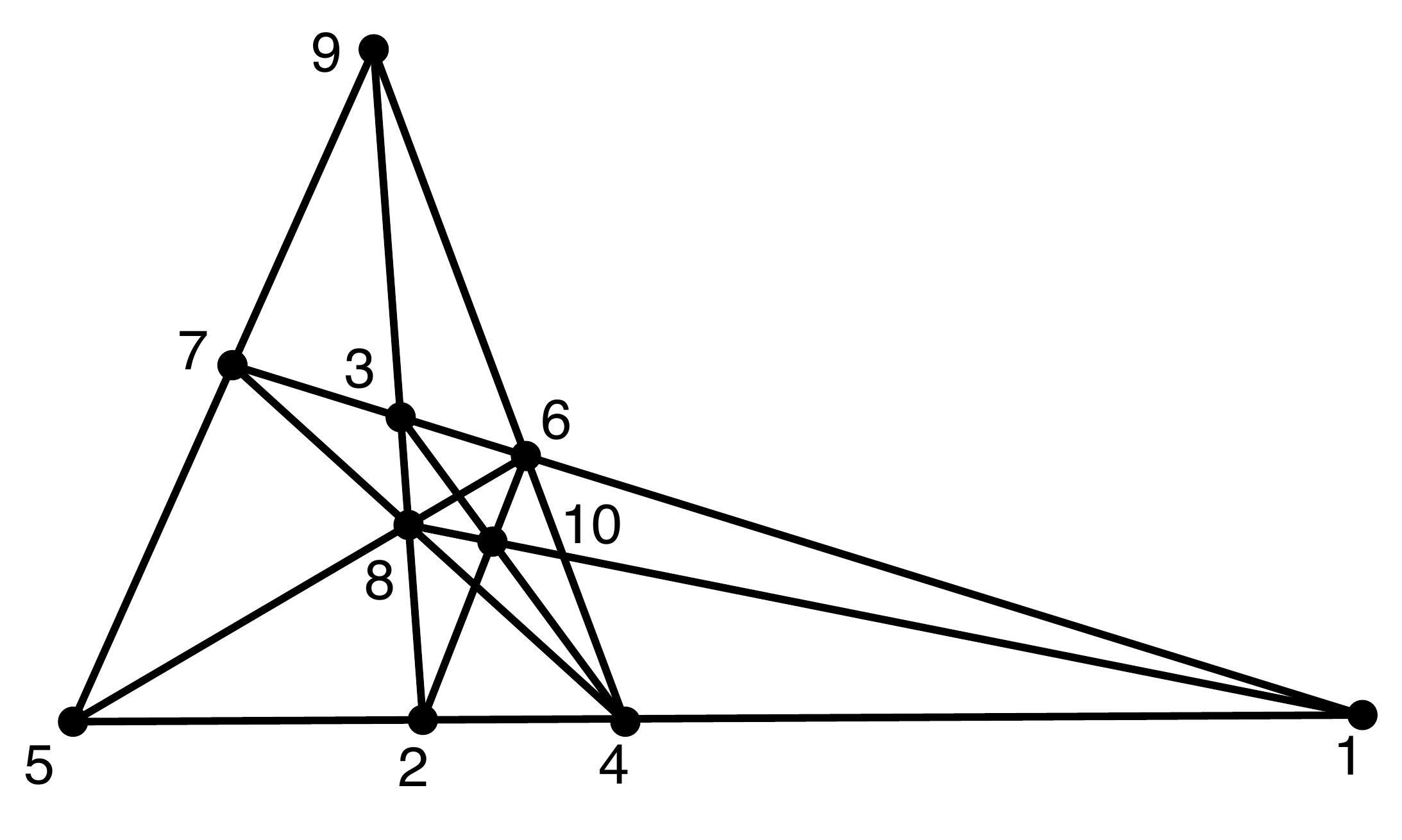}
\end{center}
\end{figure}

$$G = \frac{6(1+2t^2)}{(1+t^2)(1+4t^2)}\left[\begin{array}{ccc}
1+2t^2 & t^2 & t^2\\
t^2 & 1+2t^2 & t^2\\
t^2 & t^2 & 1+2t^2
\end{array}\right]$$

\medskip{}

$\mathcal{I}_{2}$ = $\left\{ (1,9),(2,7),(3,5),(5,10),(7,10),(9,10)\right\} $\medskip{}

$\mathcal{I}_{3}$ = $\begin{cases}
(1,8,10),(2,6,10),(3,4,10) & \nu_{31}=\frac{1+4t^{2}}{3(1+2t^{2})}\\
(4,6,9),(4,7,8),(5,6,8) & \nu_{32}=\frac{3+4t^{2}}{6(1+2t^{2})}\\
(5,7,9) & \nu_{33}=\frac{1}{2(1+2t^{2})}
\end{cases}$

\medskip{}

$\mathcal{I}_{4}$ = $\{(1,2,4,5),(1,3,6,7),(2,3,8,9)\},$ $\nu_{4}=\frac{2}{3}$\medskip{}

\subsection{$\vee$-system $(AB_{4}(t),A_{1})_{1}$}

$$\mathcal{A} =\left[\begin{array}{ccccccccccc}
1 & 2 & 3 & 4 & 5 & 6 & 7 & 8 & 9 & 10 & 11\\
\\
\sqrt{2(2t^{2}+1)} & 0 & 0 & \sqrt{2} & \sqrt{2} & t\sqrt{2} & t\sqrt{2} & t & t & t & t\\
0 & 2\sqrt{2(t^{2}+1)} & 0 & \sqrt{2} & -\sqrt{2} & 0 & 0 & 2t & -2t & 2t & -2t\\
0 & 0 & t\sqrt{\frac{2(2t^{2}-1)}{t^2+1}} & 0 & 0 & t\sqrt{2} & -t\sqrt{2} & t & t & -t & -t
\end{array}\right]$$

\medskip{}

\begin{figure} [H]
\begin{center}
\includegraphics[scale=0.5]{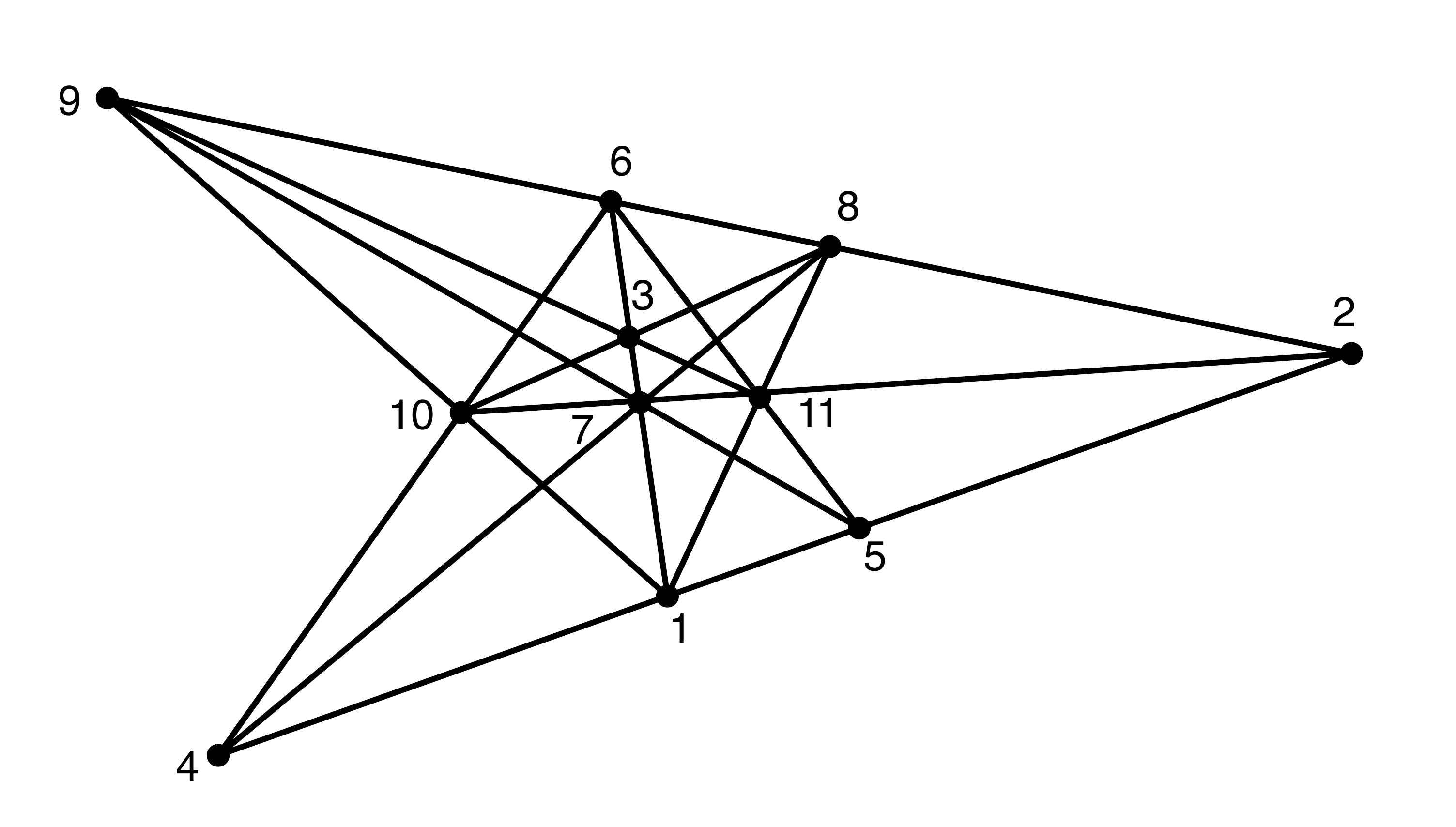}
\end{center}
\end{figure}

\medskip{}

$$G = 6\left[\begin{array}{ccc}
1+2t ^2& 0 & 0\\
0 & 2+4t^2 & 0\\
0 & 0 & \frac{t^{2}+2t^{4}}{1+t^{2}}
\end{array}\right]$$
\medskip{}

$\mathcal{I}_{2}$ = $\left\{ (2,3),(3,4),(3,5),(4,9),(4,11),(5,8),(5,10)\right\} $\medskip{}

$\mathcal{I}_{3}$ = $\begin{cases}
(4,6,10),(4,7,8),(5,6,11),(5,7,9) & \nu_{31}=\frac{3+4t^{2}}{6(1+2t^{2})}\\
(1,8,11),(1,9,10) & \nu_{32}=\frac{1+3t^{2}}{3(1+2t^{2})}\\
(3,8,10),(3,9,11) & \nu_{33}=\frac{t^{2}}{(1+2t^{2})}
\end{cases}$

\medskip{}

$\mathcal{I}_{4}$ = $\begin{cases}
(2,6,8,9),(2,7,10,11) & \nu_{41}=\frac{2}{3}\\
(1,2,4,5) & \nu_{42}=\frac{3+2t^{2}}{3(1+2t^{2})}\\
(1,3,6,7) & \nu_{43}=\frac{1+4t^{2}}{3(1+2t^{2})}
\end{cases}$\medskip{}

\subsection{$\vee$-system $G_{3}(t)$}

$$\mathcal{A} =\left[\begin{array}{ccccccccccccc}
1 & 2 & 3 & 4 & 5 & 6 & 7 & 8 & 9 & 10 & 11 & 12 & 13\\
\\
\sqrt{2t+1} & 0 & \sqrt{2t+1} & \sqrt{\frac{2t-1}{3}} & 2\sqrt{\frac{2t-1}{3}} & \sqrt{\frac{2t-1}{3}} & 0 & 1 & 1 & 0 & 0 & 1 & 1\\
0 & \sqrt{2t+1} & \sqrt{2t+1} & -\sqrt{\frac{2t-1}{3}} & \sqrt{\frac{2t-1}{3}} & 2\sqrt{\frac{2t-1}{3}} & 0 & 0 & 0 & 1 & 1 & 1 & 1\\
0 & 0 & 0 & 0 & 0 & 0 & \sqrt{\frac{3}{t}} & 1 & -1 & 1 & -1 & 1 & -1
\end{array}\right]$$\medskip{}

\begin{figure} [H]
\begin{center}
\includegraphics[scale=0.4]{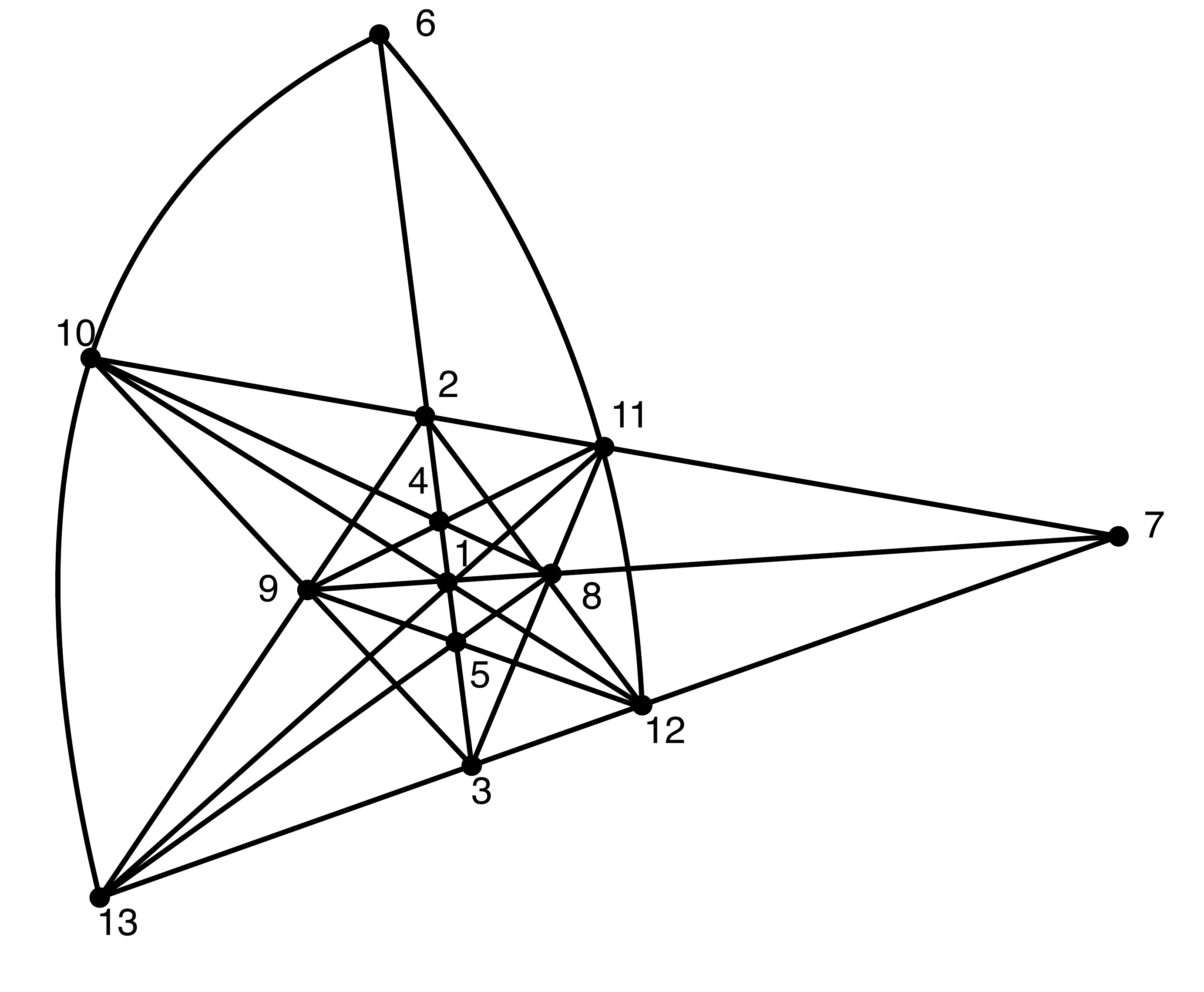}
\end{center}
\end{figure}

$$G = \left[\begin{array}{ccc}
4(1+2t) & 2(1+2t) & 0\\
2(1+2t) & 4(1+2t) & 0\\
0 & 0 & 3(2+\frac{1}{t})
\end{array}\right]$$

$\mathcal{I_{\mathrm{2}}}$ =$\left\{ (4,7),(4,12),(4,13),(5,7),(5,10),(5,11),(6,7),(6,8),(6,9)\right\} $

\medskip{}

$\mathcal{I}_{3}$ = $\begin{cases}
(2,9,13),(2,8,12),(3,8,11),(1,11,13),(1,10,12),(3,9,10) & \nu_{31}=\frac{3+4t}{6(1+2t)}\\
(4,8,10),(6,11,12),(6,10,13),(5,9,12),(5,8,13),(4,9,11) & \nu_{32}=\frac{1+4t}{6(1+2t)}
\end{cases}$

\medskip{}

$\mathcal{I}_{4}$ =\{(2,7,10,11),(1,7,8,9),(3,7,12,13)\} ,$\nu_{4}=\frac{3+2t}{3+6t}$

\medskip{}

$\mathcal{I}_{6}$ = $\{(1,2,3,4,5,6)\},$$\nu_{6}=\frac{2t}{1+2t}$\medskip{}

\subsection{$\vee$-system $(E_{7},A_{1}^{2}\times A_{2})$}

\medskip{}
\begin{eqnarray*}
\mathcal{A}
 & =
\left[\begin{array}{ccccccccccccc}
1 & 2 & 3 & 4 & 5 & 6 & 7 & 8 & 9 & 10 & 11 & 12 & 13\\
\\
\sqrt{3} & \sqrt{3} & 2 & 0 & 0 & \frac{1}{\sqrt{2}} & -\frac{1}{\sqrt{2}} & -\frac{1}{\sqrt{2}} & \frac{1}{\sqrt{2}} & \sqrt{\frac{3}{2}} & -\sqrt{\frac{3}{2}} & -\sqrt{\frac{3}{2}} & \sqrt{\frac{3}{2}}\\
\sqrt{3} & -\sqrt{3} & 0 & 2\sqrt{6} & 0 & \frac{3}{\sqrt{2}} & -\frac{3}{\sqrt{2}} & \frac{3}{\sqrt{2}} & -\frac{3}{\sqrt{2}} & \sqrt{\frac{3}{2}} & -\sqrt{\frac{3}{2}} & \sqrt{\frac{3}{2}} & -\sqrt{\frac{3}{2}}\\
0 & 0 & 0 & 0 & 1 & \frac{1}{\sqrt{2}} & \frac{1}{\sqrt{2}} & \frac{1}{\sqrt{2}} & \frac{1}{\sqrt{2}} & \sqrt{\frac{3}{2}} & \sqrt{\frac{3}{2}} & \sqrt{\frac{3}{2}} & \sqrt{\frac{3}{2}}
\end{array}\right]
\end{eqnarray*}
\medskip{}

\begin{figure} [H]
\begin{center}
\includegraphics[scale=0.56]{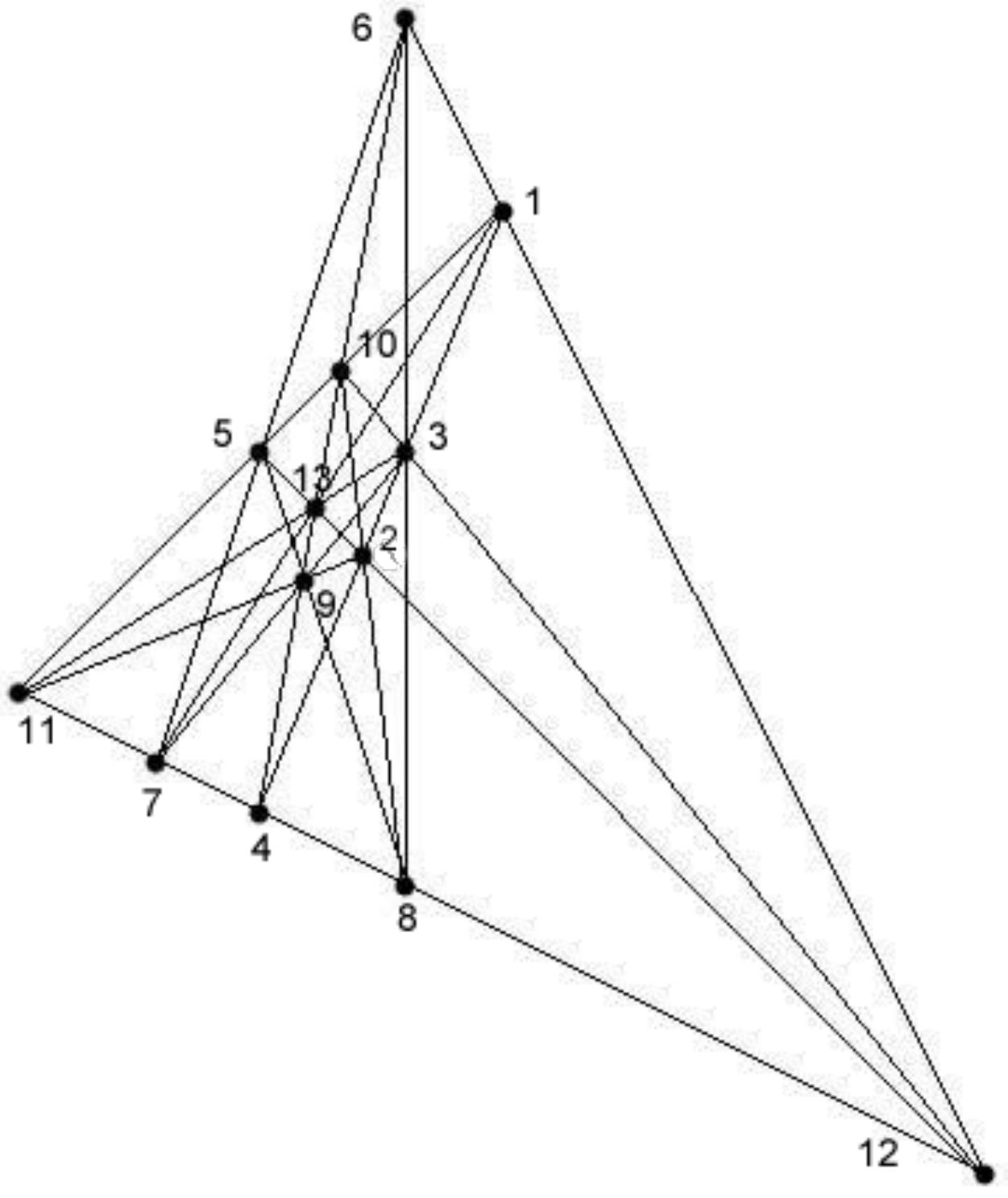}
\end{center}
\end{figure}

$$G = 9\left[\begin{array}{ccc}
2 & 0 & 0\\
0 & 6 & 0\\
0 & 0 & 1
\end{array}\right]$$

$\mathcal{I_{\mathrm{2}}}$ =$\left\{ (1,8),(1,9),(2,6),(2,7),(3,5),(4,5),(6,11),(7,10),(8,13),(9,12)\right\} $\medskip{}

$\mathcal{I_{\mathrm{3}}}$ =$\begin{cases}
(5,9,8),(7,5,6) & \nu_{31}=\frac{2}{9}\\
(10,3,12),(11,13,3) & \nu_{32}=\frac{7}{18}\\
(6,3,8),(7,9,3) & \nu_{33}=\frac{5}{18}\\
(6,1,12),(10,2,8),(7,13,1),(11,9,2) & \nu_{34}=\frac{1}{3}
\end{cases}$

\medskip{}

$\mathcal{I_{\mathrm{4}}}$ =$\begin{cases}
(5,13,2,12),(11,5,10,1) & \nu_{41}=\frac{4}{9}\\
(4,2,3,1) & \nu_{42}=\frac{5}{9}
\end{cases}$

\medskip{}

$\mathcal{I_{\mathrm{5}}}$ =$\left\{ (11,7,4,8,12),(6,10,13,9,4)\right\} ,\nu_{5}=\frac{2}{3}$\medskip{}

\subsection{$\vee$-system $F_{3}(t)$}

\medskip{}

\begin{eqnarray*}
\mathcal{A}
 & =
\left[\begin{array}{ccccccccccccc}
1 & 2 & 3 & 4 & 5 & 6 & 7 & 8 & 9 & 10 & 11 & 12 & 13\\
\\
\sqrt{4t^{2}+2} & 0 & 0 & 1 & 1 & 1 & 1 & 0 & 0 & t\sqrt{2} & t\sqrt{2} & t\sqrt{2} & t\sqrt{2}\\
0 & \sqrt{4t^{2}+2} & 0 & 1 & -1 & 0 & 0 & 1 & 1 & t\sqrt{2} & -t\sqrt{2} & t\sqrt{2} & -t\sqrt{2}\\
0 & 0 & \sqrt{4t^{2}+2} & 0 & 0 & 1 & -1 & 1 & -1 & t\sqrt{2} & t\sqrt{2} & -t\sqrt{2} & -t\sqrt{2}
\end{array}\right]
\end{eqnarray*}

\medskip{}

\begin{figure} [H]
\begin{center}
\includegraphics[scale=0.5]{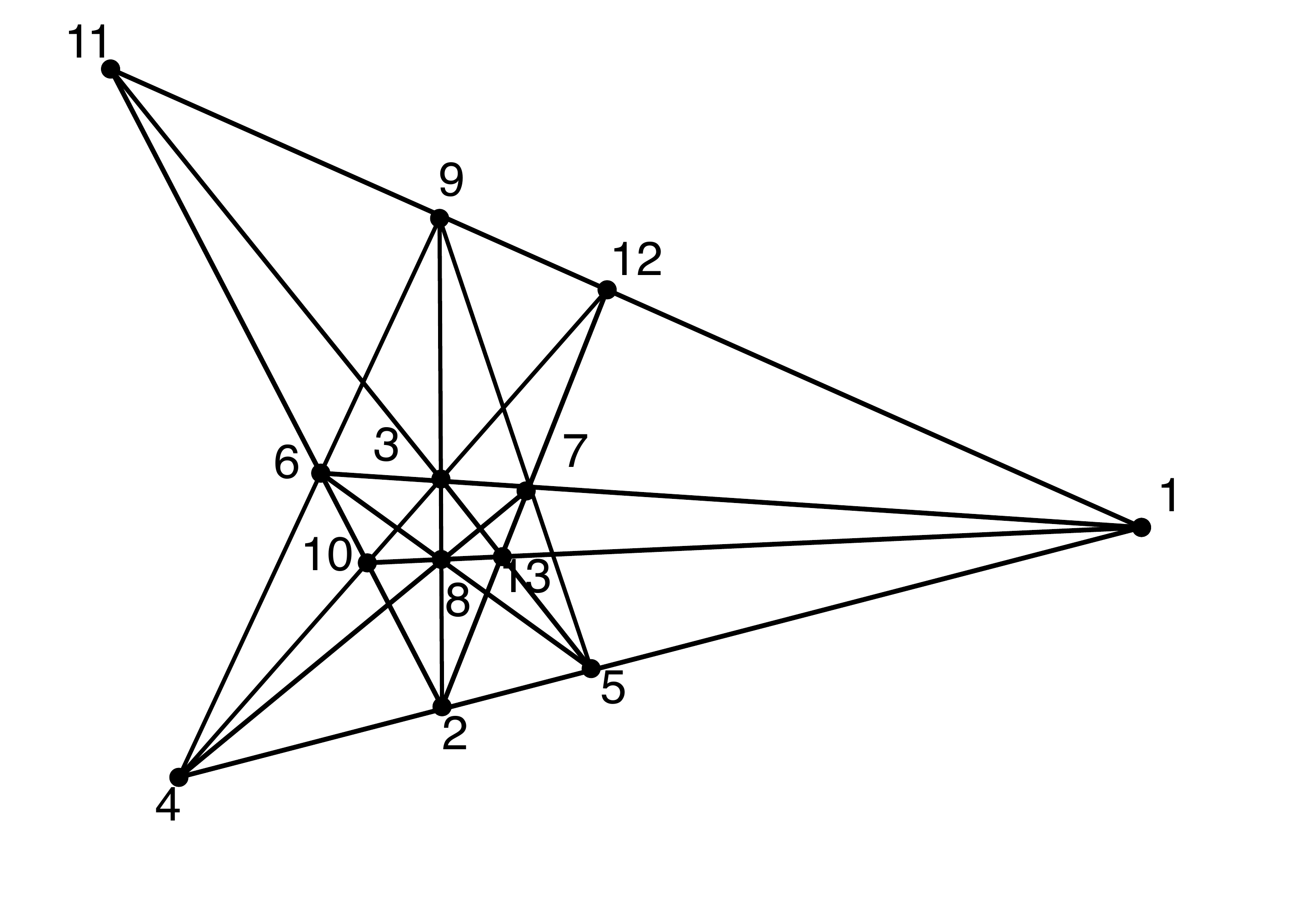}
\end{center}
\end{figure} 
$$G = (6+12t^2)I$$
\medskip

$\mathcal{I_{\mathrm{2}}}$ =\{(4,11),(4,13),(5,10),(5,12),(6,13),(7,10),(7,11),(8,11),(8,12),(9,10),(9,13)\}

\medskip{}

$\mathcal{I_{\mathrm{3}}}$ =$\begin{cases}
(4,6,9),(4,7,8),(5,6,8),(5,7,9), & \nu_{31}=\frac{1}{2+4t^{2}}
\end{cases}$\medskip{}

$\mathcal{I_{\mathrm{4}}}$ =$\begin{cases}
(1,2,4,5),(1,3,6,7),(2,3,8,9), & \nu_{32}=\frac{2(1+t^{2})}{3+6t^{2}}\\
(1,8,10,13),(1,9,11,12),(2,6,10,11),(2,7,12,13),\\
(3,4,10,12),(3,5,11,13), & \nu_{33}=\frac{1+4t^{2}}{3+6t^{2}}
\end{cases}$\medskip{}

\subsection{Coxeter $\vee$-system $H_{3}$}

$$\mathcal{A} = \left[\begin{array}{ccccccccccccccc}
1 & 2 & 3 & 4 & 5 & 6 & 7 & 8 & 9 & 10 & 11 & 12 & 13 & 14 & 15\\
2\phi & 0 & 0 & 1 & -1 & 1 & 1 & \phi & -\phi & \phi & \phi & \phi^{2} & \phi^{2} & -\phi^{2} & \phi^{2}\\
0 & 2\phi & 0 & \phi & \phi & -\phi & \phi & -\phi^{2} & \phi^{2} & \phi^{2} & \phi^{2} & 1 & -1 & 1 & 1\\
0 & 0 & 2\phi & \phi^{2} & \phi^{2} & \phi^{2} & -\phi^{2} & 1 & 1 & -1 & 1 & -\phi & \phi & \phi & \phi
\end{array}\right]$$
where $\phi$ is the golden ratio $\phi=\frac{1+\sqrt{5}}{2}.$ 
\medskip{}

\begin{figure} [H]
\begin{center}
\includegraphics[scale=0.4]{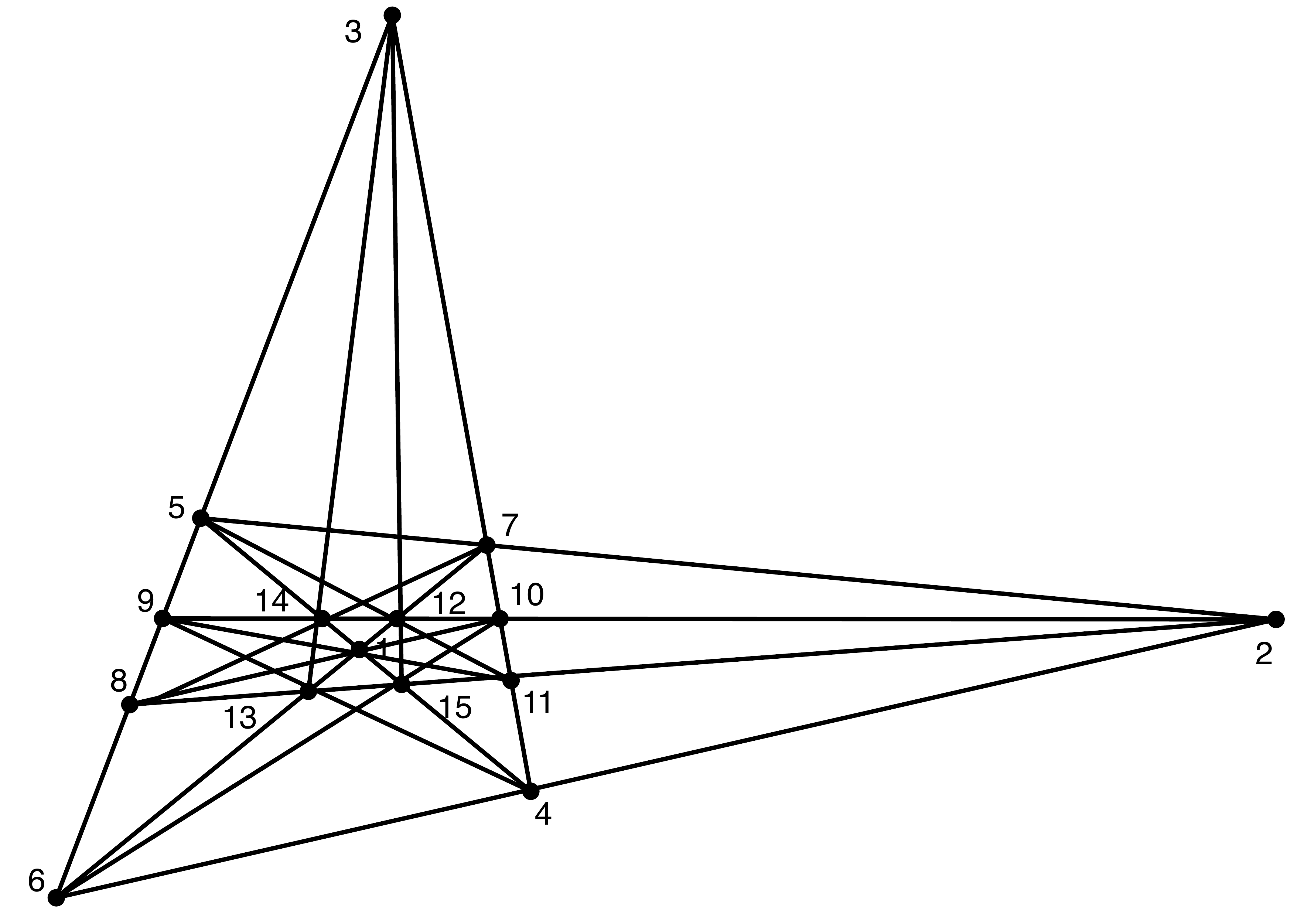}
\end{center}
\end{figure}

$$G = 10(3+\sqrt{5})\mathrm{I}$$

$\mathcal{I_{\mathrm{2}}}$ =$\begin{cases}(1,2),(1,3),(2,3),(4,8),(4,12),(5,10),(5,13),(6,11),\\
(6,14),(7,9),(7,15),(8,12),(9,15),(10,13),(11,14).
\end{cases}$ 

\medskip{}

$\mathcal{I_{\mathrm{3}}}$ =$\begin{cases}
(1,8,10),(1,9,11),(2,4,6),(2,5,7),(3,12,15), \\ 
(3,13,14),(4,9,13),(5,11,12),(6,10,15),(7,8,14) & \nu_{31}=\frac{3}{10}
\end{cases}$ 

\medskip{}

$\mathcal{I_{\mathrm{5}}}$ =$\begin{cases}
(1, 4, 5, 14, 15), (1, 6, 7, 12, 13), (2, 8,11, 13, 15), (2, 9, 10, 12, 14),\\
(3,4,7,10,11),(3,5,6,8,9) & \nu_{5}=\frac{1}{2}
\end{cases}$ 

\subsection{$\vee$-system $(E_{8},A_{1}\times A_{4})$}

\begin{eqnarray*}
\mathcal{A}
 & =
 & \left[\begin{array}{cccccccccccc}
1 & 2 & 3 & 4 & 5 & 6 & 7 & 8 & 9 & 10 & 11 & 12\\
\\
\sqrt{10} & \sqrt{10} & \sqrt{2} & \sqrt{2} & 0 & 0 & 2 & 0 & 1 & -1 & \sqrt{5} & -\sqrt{5}\\
\sqrt{10} & -\sqrt{10} & 0 & 0 & \sqrt{5} & \sqrt{10} & 0 & 2\sqrt{10} & \frac{2}{5} & \frac{5}{2} & -\frac{3\sqrt{5}}{2} & -\frac{3\sqrt{5}}{2}\\
0 & 0 & \sqrt{2} & -\sqrt{2} & \sqrt{5} & -\sqrt{10} & 0 & 0 & \frac{1}{2} & \frac{1}{2} & \frac{\sqrt{5}}{2} & \frac{\sqrt{5}}{2}
\end{array}\cdots\right.\\
 & \begin{array}{cccccc}
\end{array} & \left.\cdots\begin{array}{cccc}
13 & 14 & 15 & 16\\
\\
\sqrt{10} & -\sqrt{10} & 0 & 0\\
\frac{\sqrt{10}}{2} & \frac{\sqrt{10}}{2} & -\frac{5}{\sqrt{2}} & \frac{3\sqrt{10}}{2}\\
\frac{\sqrt{10}}{2} & \frac{\sqrt{10}}{2} & \frac{1}{\sqrt{2}} & \frac{\sqrt{10}}{2}
\end{array}\right]
\end{eqnarray*}
\medskip{}

\begin{figure} [H]
\begin{center}
\includegraphics[scale=0.9]{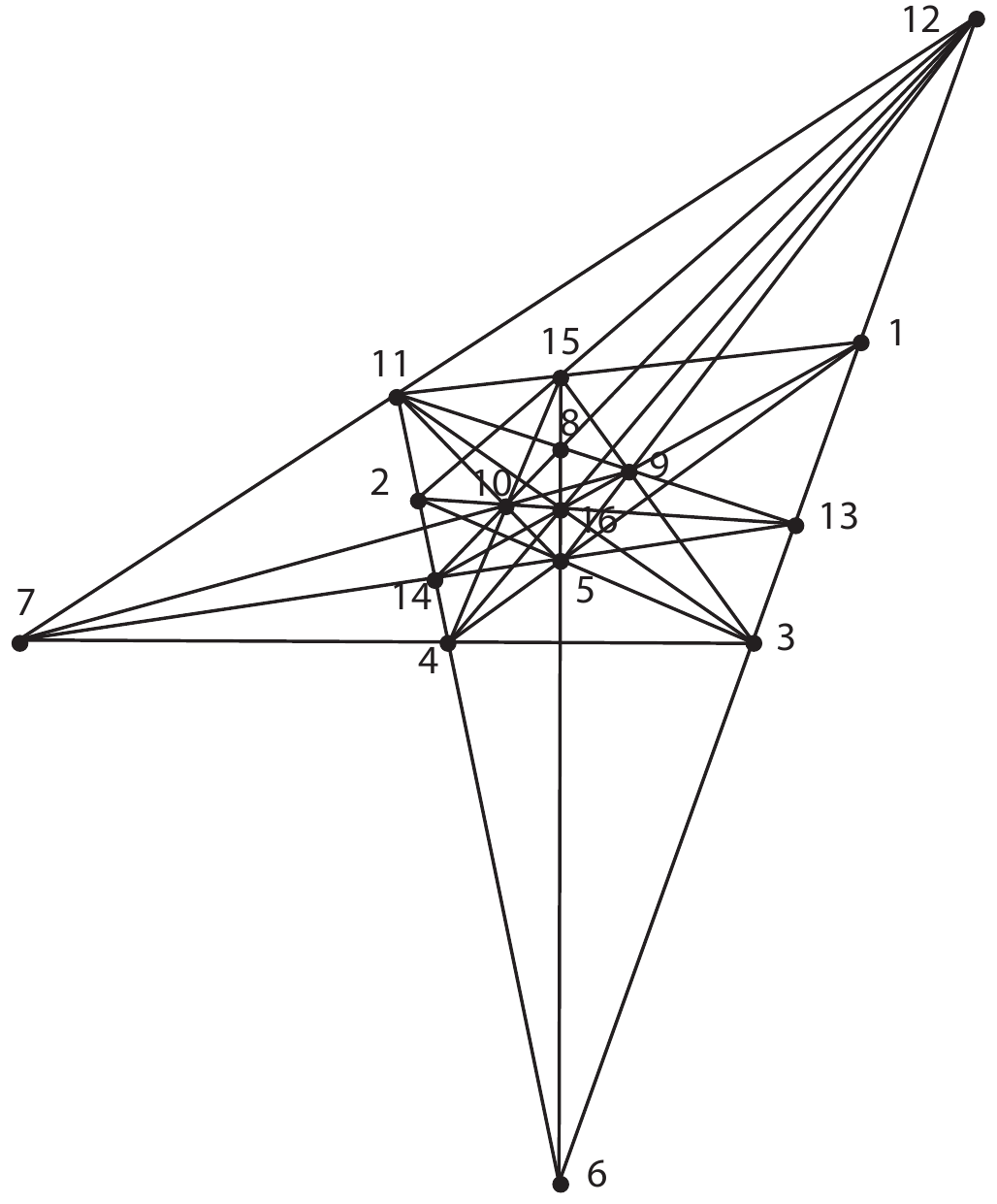}
\end{center}
\end{figure}
\medskip{}

$$G = 30 \left[\begin{array}{ccc}
2 & 0 & 0\\
0 & 5 & 0\\
0 & 0 & 1
\end{array}\right]$$

\medskip{}

$\mathcal{I_{\mathrm{2}}}$ =$\{(1,10),(2,9),(3,8),(3,10),(3,14),(4,8),(4,9),(4,13),(6,7),(6,9),(6,10),$

\qquad{}$(7,15),(7,16),(13,15),(14,15)\}$\medskip{}

$\mathcal{I_{\mathrm{3}}}$ =$\begin{cases}
(5,10,11),(5,9,12),(7,11,12) & \nu_{31}=\frac{7}{30}\\
(7,9,10) & \nu_{32}=\frac{1}{10}\\
(3,16,11),(4,5,1),(4,16,12),(3,5,2),(11,15,1),(2,15,12) & \nu_{33}=\frac{4}{15}\\
(4,10,15),(3,9,15),(3,4,7) & \nu_{34}=\frac{2}{15}
\end{cases}$

\medskip{}

$\mathcal{I_{\mathrm{4}}}=\{(7,2,8,1),(12,8,10,14),(14,16,9,1),(13,16,10,2),(13,5,14,7),(13,9,8,11)\}$,
$ $

\hfill{}$\nu_{4}=\frac{2}{5}$

\medskip{}

$\mathcal{I_{\mathrm{5}}}$ =$\{(15,8,16,5,6),(12,1,13,3,6),(11,2,14,4,6)\},\nu_{5}=\frac{3}{5}$

\medskip{}

\subsection{$\vee$-system $(E_{8},A_{2}\times A_{3})$}

\begin{eqnarray*}
\mathcal{A}
 & =
 & \left[\begin{array}{ccccccccccccc}
1 & 2 & 3 & 4 & 5 & 6 & 7 & 8 & 9 & 10 & 11 & 12 & 13\\
\\
2\sqrt{3} & 2\sqrt{3} & 0 & 0 & \sqrt{3} & \sqrt{\frac{15}{2}} & 2\sqrt{3} & 0 & \frac{3}{2} & \frac{3}{2} & -3 & -3 & \frac{3\sqrt{6}}{2}\\
2\sqrt{3} & -2\sqrt{3} & 2 & 2 & 0 & 0 & 0 & 2\sqrt{6} & 2 & -2 & 2 & -2 & 0\\
0 & 0 & 2 & -2 & \sqrt{3} & -\sqrt{\frac{15}{2}} & 0 & 0 & \frac{1}{2} & \frac{1}{2} & 1 & 1 & \frac{\sqrt{6}}{2}
\end{array}\cdots\right.\\
 & \begin{array}{cccccc}
\end{array} & \left.\cdots\begin{array}{cccc}
14 & 15 & 16 & 17\\
\\
-\frac{\sqrt{3}}{2} & -\frac{\sqrt{3}}{2} & \sqrt{3} & \sqrt{3}\\
2\sqrt{3} & -2\sqrt{3} & 2\sqrt{3} & -2\sqrt{3}\\
\frac{\sqrt{3}}{2} & -\frac{\sqrt{3}}{2} & \sqrt{3} & \sqrt{3}
\end{array}\right]
\end{eqnarray*}
\medskip{}

\begin{figure} [H]
\begin{center}
\includegraphics[scale=0.95]{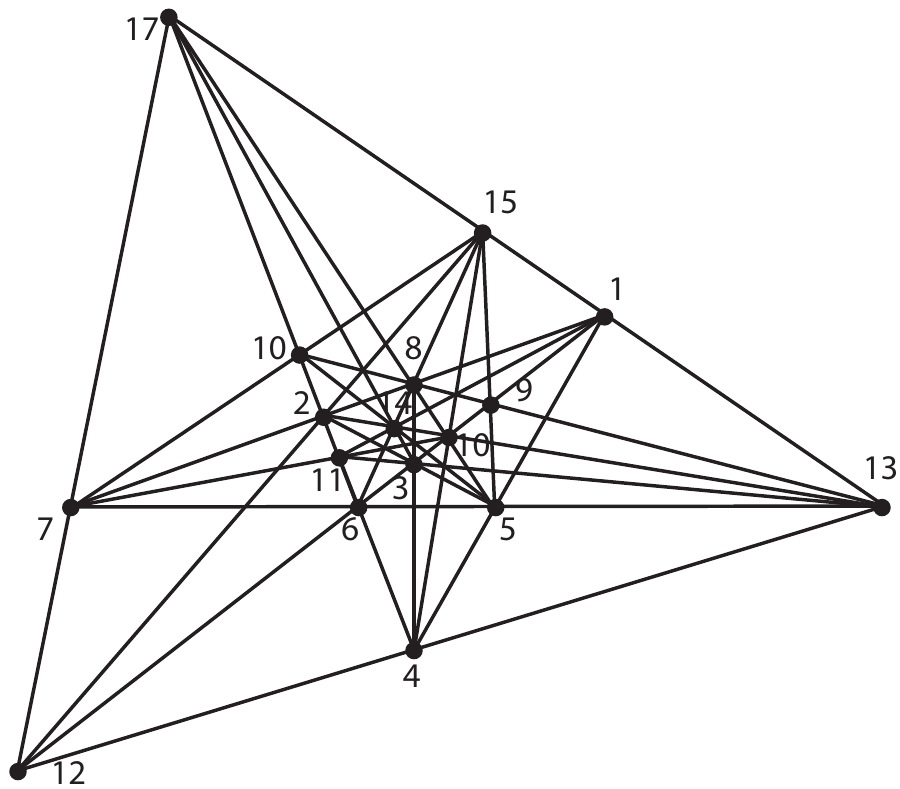}
\end{center}
\end{figure}
\medskip{}

$$G =30 \left[\begin{array}{ccc}
3 & 0 & 0\\
0 & 4 & 0\\
0 & 0 & 1
\end{array}\right]$$

$\mathcal{I_{\mathrm{2}}}$ =$\{(1,10),(2,9),(3,7),(3,10),(3,15),(4,7),(4,9),(4,14),$

\qquad{}$(5,11),(5,12),(9,11),(9,17),(10,12),(10,16),(11,15),(12,14)\}$\medskip{}

$\mathcal{I_{\mathrm{3}}}$ =$\begin{cases}
\{(5,9,15),(5,10,14),(7,9,14),(7,10,15)\} & \nu_{31}=\frac{1}{6}\\
\{(3,14,17),(4,12,13),(4,15,16),(7,11,16),(7,12,17),(8,11,12),\\
(1,4,5),(1,11,14),(2,3,5),(2,12,15),(3,4,8),(3,11,13)\} & \nu_{32}=\frac{4}{15}
\end{cases}$ 

\medskip{}

$\mathcal{I_{\mathrm{4}}}$ =$\begin{cases}
(8,9,10,13) & \nu_{41}=\frac{4}{15}\\
(1,2,7,8),(1,13,15,17),(2,13,14,16),(5,6,7,13),(5,8,16,17),(6,8,14,15) & \nu_{42}=\frac{2}{5}
\end{cases}$

\medskip{}

$\mathcal{I_{\mathrm{6}}}$ =$\{(1,3,6,9,12,16),(2,4,6,10,11,17)\}$,
$\nu_{6}=\frac{3}{5}$

\bigskip\bigskip

\subsection{$\vee$-system $(E_{8},A_{1}^{2}\times A_{3})$}

\begin{eqnarray*}
\mathcal{A}
 & =
 & \left[\begin{array}{ccccccccccccccc}
1 & 2 & 3 & 4 & 5 & 6 & 7 & 8 & 9 & 10 & 11 & 12 & 13\\
\\
2 & 2 & 0 & 0 & 2 & 2 & 2 & 0 & 0 & \frac{\sqrt{2}}{2} & \frac{\sqrt{2}}{2} & \frac{\sqrt{2}}{2} & \frac{\sqrt{2}}{2}\\
2 & -2 & 2 & 2 & 0 & 0 & 0 & 2\sqrt{10} & 0 & 2\sqrt{2} & -2\sqrt{2} & -2\sqrt{2} & 2\sqrt{2}\\
0 & 0 & 2 & -2 & 2 & -2 & 0 & 0 & 2 & \frac{\sqrt{2}}{2} & -\frac{\sqrt{2}}{2} & \frac{\sqrt{2}}{2} & -\frac{\sqrt{2}}{2}
\end{array}\cdots\right.\\
 & \begin{array}{cccccc}
\end{array} & \left.\cdots\begin{array}{cccc}
14 & 15 & 16 & 17\\
\\
\sqrt{2} & \sqrt{2} & \sqrt{2} & \sqrt{2}\\
2\sqrt{2} & -2\sqrt{2} & -2\sqrt{2} & 2\sqrt{2}\\
\sqrt{2} & -\sqrt{2} & \sqrt{2} & -\sqrt{2}
\end{array}\right]
\end{eqnarray*}
\medskip{}

\begin{figure} [H]
\begin{center}
\includegraphics[scale=0.82]{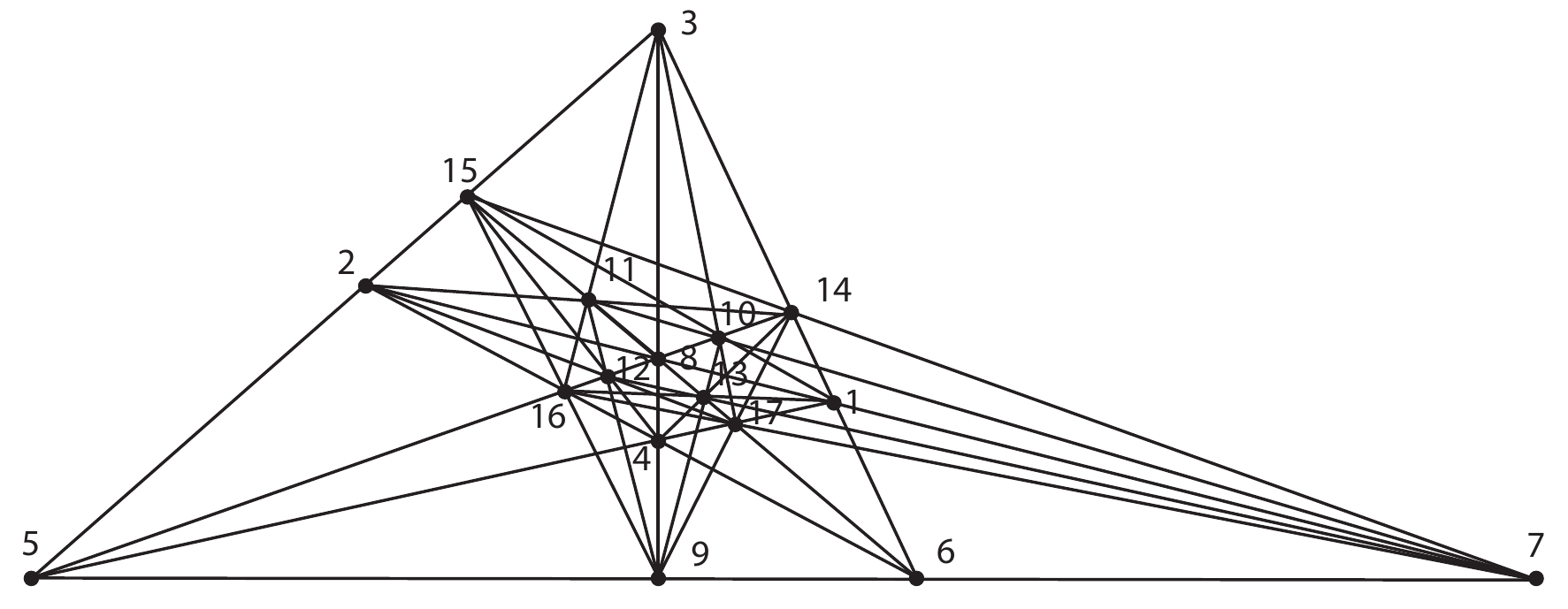}
\end{center}
\end{figure}

$$G = 30\left[\begin{array}{ccc}
1 & 0 & 0\\
0 & 4 & 0\\
0 & 0 & 1
\end{array}\right]$$
\medskip{}

$\mathcal{I_{\mathrm{2}}}$ =$\{(1,9),(1,11),(1,12),(2,9),(2,10),(2,13),(3,7),(3,12),$

\qquad{}$(3,13),(4,7),(4,10),(4,11),(5,11),(5,13),(6,10),(6,12)\}$\medskip{}

$\mathcal{I_{\mathrm{3}}}$ =$\begin{cases}
(7,10,11),(7,12,13),(9,10,13),(9,11,12) & \nu_{31}=\frac{1}{6}\\
(9,15,16),(9,14,17),(7,16,17),(7,14,15) & \nu_{32}=\frac{4}{15}\\
(1,10,15),(1,13,16),(2,11,14),(2,12,17),(3,10,17),\\
(3,11,16),(4,12,15),(4,13,14) & \nu_{33}=\frac{7}{30}
\end{cases}$ 

\medskip{}

$\mathcal{I_{\mathrm{4}}}$ =$\{(1,2,7,8),(1,3,6,14),(1,4,5,17),(2,3,5,15),(2,4,6,\ 16),(3,4,8,9),(5,6,7,9)\}$,

\hfill{}$\nu_{4}=\frac{2}{5}$

\medskip{}

$\mathcal{I_{\mathrm{6}}}$ =$\{(6,8,11,13,15,17),(5,8,10,12,14,16)\}$,$\nu_{6}=\frac{3}{5}$

\medskip{}

\subsection{$\vee$-system $(E_{8},A_{1}^{3}\times A_{2})$}

\begin{eqnarray*}
\mathcal{A} & = & \left[\begin{array}{ccccccccccccccc}
1 & 2 & 3 & 4 & 5 & 6 & 7 & 8 & 9 & 10 & 11 & 12 & 13 & 14 & 15\\
\\
\sqrt{3} & 3 & 0 & 1 & \sqrt{3} & 0 & \sqrt{6} & 0 & 0 & \sqrt{6} & \sqrt{3} & 3 & 1 & \sqrt{3} & 0\\
\sqrt{3} & 0 & 3 & -1 & 0 & \sqrt{3} & 0 & \sqrt{6} & 0 & \sqrt{6} & \sqrt{3} & 3 & 2 & 0 & \sqrt{3}\\
0 & 3 & 3 & 0 & -\sqrt{3} & -\sqrt{3} & 0 & 0 & 6\sqrt{2} & 3\sqrt{6} & 4\sqrt{3} & 6 & 3 & 3\sqrt{3} & 3\sqrt{3}
\end{array}\cdots\right.\\
 & \begin{array}{cccccc}
\end{array} & \left.\cdots\begin{array}{cccc}
16 & 17 & 18 & 19\\
\\
2 & \sqrt{6} & 0 & \sqrt{6}\\
1 & 0 & \sqrt{6} & \sqrt{6}\\
3 & 2\sqrt{6} & 2\sqrt{6} & \sqrt{6}
\end{array}\right]
\end{eqnarray*}

\begin{figure} [H]
\begin{center}
\includegraphics[scale=0.75]{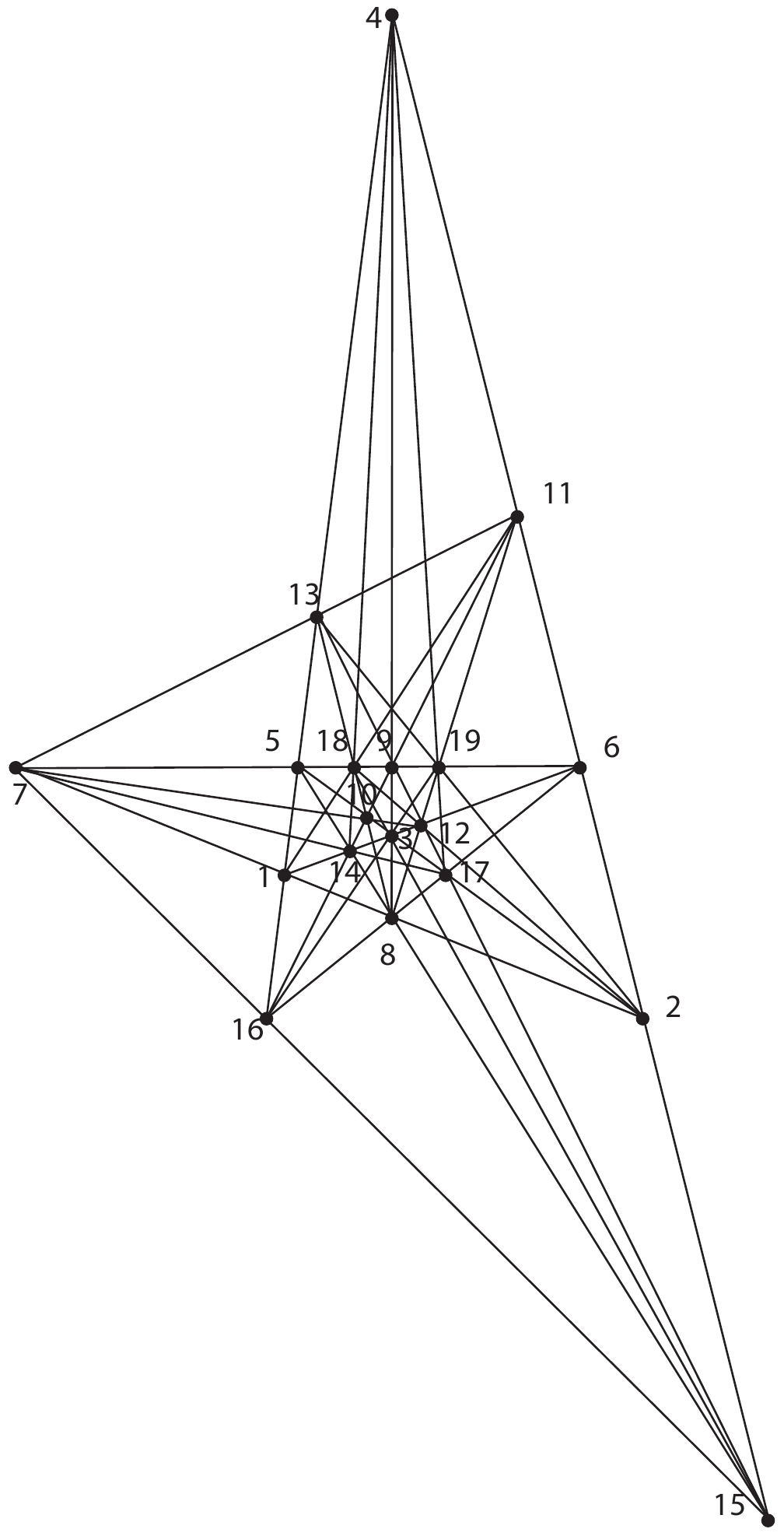}
\end{center}
\end{figure}
\medskip{}

$$G = 30\left[\begin{array}{ccc}
1 & 0 & 0\\
0 & 2 & 0\\
0 & 0 & 3
\end{array}\right]$$
\medskip{}

\[
\mathcal{I_{\mathrm{2}}}=\{(1,9),(1,15),(1,17),(2,9),(2,14),(2,16),(3,7),(3,11),
\]
\[
(3,13),(4,7),(4,10),(4,12),(5,11),(5,12),(6,10),(6,13),(10,15),(11,17),(12,16),
\]
\[
(13,14),(14,19),(15,19),(16,18),(17,18)\}
\]
\[
\]

\medskip{}

$\mathcal{I_{\mathrm{3}}}$ =$\begin{cases}
(3,15,18),(3,16,19),(4,14,18),(4,17,19),(7,14,17),(7,15,16) & \nu_{31}=\frac{7}{30}\\
(7,13,11),(2,12,18),(2,13,19),(1,10,19),(1,11,18),(7,10,12) & \nu_{32}=\frac{1}{6}
\end{cases}$ 

\medskip{}

$\mathcal{I_{\mathrm{4}}}$ =$\begin{cases}
(1,2,7,8),(8,10,13,18),(8,11,12,19) & \nu_{31}=\frac{4}{15}\\
(3,4,8,9),(5,8,14,15),(6,8,16,17) & \nu_{32}=\frac{2}{5}
\end{cases}$

\medskip{}

$\mathcal{I_{\mathrm{5}}}$ =$\begin{cases}
(9,12,13,15,17),(9,10,11,14,16),(2,4,6,11,15),\\
(2,3,5,10,17),(1,4,5,13,16),(1,3,6,12,14)& \nu_{5}=\frac{2}{5}
\end{cases}$

\medskip{}

$\mathcal{I_{\mathrm{6}}}$ =$\{(5,6,7,9,18,19)\},$$\nu_{6}=\frac{3}{5}$

\subsection{$\vee$-system $(E_{8},A_{2}^{2}\times A_{1})$}

\begin{eqnarray*}
\mathcal{A}
 & =
 & \left[\begin{array}{ccccccccccccccc}
1 & 2 & 3 & 4 & 5 & 6 & 7 & 8 & 9 & 10 & 11 & 12 & 13 & 14 & 15\\
\\
\sqrt{3} & 3 & 0 & 1 & \sqrt{3} & 0 & \sqrt{6} & 0 & 0 & \sqrt{6} & \sqrt{3} & 3 & 1 & \sqrt{3} & 0\\
\sqrt{3} & 0 & 3 & -1 & 0 & \sqrt{3} & 0 & \sqrt{6} & 0 & \sqrt{6} & \sqrt{3} & 3 & 2 & 0 & \sqrt{3}\\
0 & 3 & 3 & 0 & -\sqrt{3} & -\sqrt{3} & 0 & 0 & 6\sqrt{2} & 3\sqrt{6} & 4\sqrt{3} & 6 & 3 & 3\sqrt{3} & 3\sqrt{3}
\end{array}\cdots\right.\\
 & \begin{array}{cccccc}
\end{array} & \left.\cdots\begin{array}{cccc}
16 & 17 & 18 & 19\\
\\
2 & \sqrt{6} & 0 & \sqrt{6}\\
1 & 0 & \sqrt{6} & \sqrt{6}\\
3 & 2\sqrt{6} & 2\sqrt{6} & \sqrt{6}
\end{array}\right]
\end{eqnarray*}
\medskip{}

\begin{figure} [H]
\begin{center}
\includegraphics[scale=0.6]{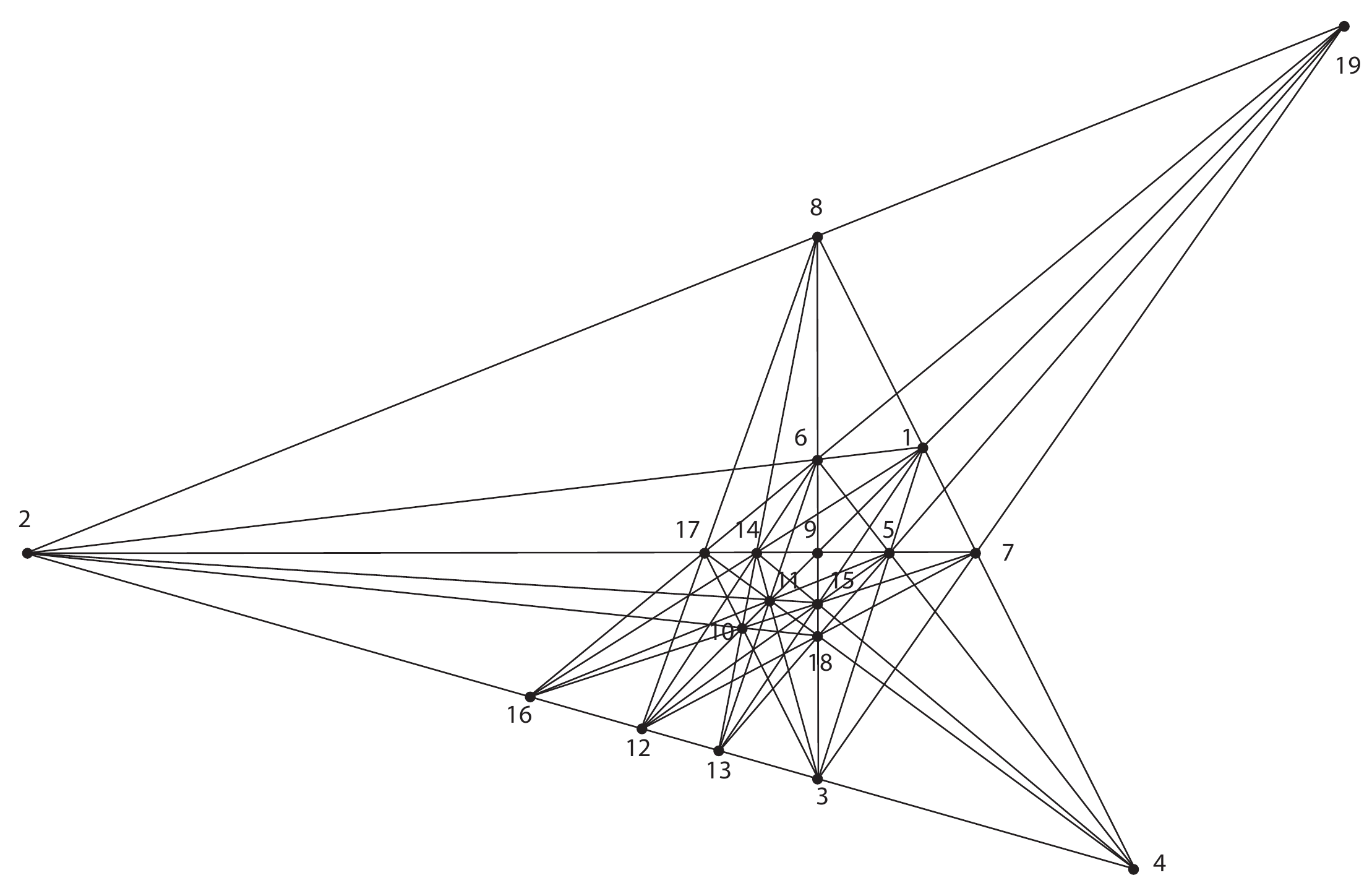}
\end{center}
\end{figure}

$$G = 30\left[\begin{array}{ccc}
2 & 1 & 3\\
1 & 2 & 3\\
3 & 3 & 12
\end{array}\right]$$
\medskip{}

\[
\mathcal{I_{\mathrm{2}}}=\{(1,17),(1,18),(4,9),(4,10),(4,19),(5,8),(5,10),(6,7),(6,10),(7,11),(7,13),
\]
\[
(8,11),(8,16),(9,13),(9,16),(13,17),(14,18),(14,19),(15,17),(15,19),(16,18)\}
\]
\[
\]

\medskip{}

$\mathcal{I_{\mathrm{3}}}$ =$\begin{cases}
(1,2,6),(1,3,5),(2,11,15),(3,11,14),(5,12,15),(6,12,14) & \nu_{31}=\frac{7}{30}\\
(7,12,18),(8,12,17),(3,10,17),(3,7,19),(2,10,18),(2,8,19) & \nu_{32}=\frac{4}{15}\\
(1,13,15),(1,14,16),(4,5,6),(4,14,15),(5,11,16),(6,11,13) & \nu_{33}=\frac{1}{6}
\end{cases}$ 

\medskip{}

$\mathcal{I_{\mathrm{4}}}$ =$\begin{cases}
(1,4,7,8),(4,11,17,18),(5,13,18,19),(6,16,17,19),\\
(7, 10,15,16),(8,10,13,14) & \nu_{4}=\frac{4}{15}\end{cases}$

\medskip{}

$\mathcal{I_{\mathrm{6}}}$ =$\begin{cases}
(2,3,4,12,13,16); & \nu_{61}=\frac{2}{5}\\
(2,5,7,9,14,17),(1,9,10,11,12,19),(3,6,8,9,15,18) & \nu_{62}=\frac{3}{5}
\end{cases}$

\medskip{}

\subsection{$\vee$-system $(H_{4},A_{1})$}

\begin{eqnarray*}
\mathcal{A}
 & =
 & \left[\begin{array}{cccccccccccccccc}
1 & 2 & 3 & 4 & 5 & 6 & 7 & 8 & 9 & 10 & 11 & 12 & 13 & 14 & 15 & 16\\
\\
1 & 0 & 0 & \frac{\sqrt{2}}{2} & \frac{\sqrt{2}}{2} & \frac{\sqrt{2}}{2} & \frac{\sqrt{2}}{2} & a & a & a & a & b & b & b & b & \frac{1}{2}\\
0 & 1 & 0 & \frac{\sqrt{2}}{2} & \frac{\sqrt{2}}{2} & -\frac{\sqrt{2}}{2} & -\frac{\sqrt{2}}{2} & \frac{1}{2} & \frac{1}{2} & -\frac{1}{2} & -\frac{1}{2} & a & a & -a & -a & b\\
0 & 0 & 1 & \frac{\sqrt{2}}{2} & -\frac{\sqrt{2}}{2} & \frac{\sqrt{2}}{2} & -\frac{\sqrt{2}}{2} & b & -b & b & -b & \frac{1}{2} & -\frac{1}{2} & \frac{1}{2} & -\frac{1}{2} & a
\end{array}\cdots\right.\\
 &  & ...\begin{array}{ccccccccccc}
17 & 18 & 19 & 20 & 21 & 22 & 23 & 24 & 25 & 26 & 27\\
\\
\frac{1}{2} & \frac{1}{2} & \frac{1}{2} & a\sqrt{2} & a\sqrt{2} & 0 & 0 & b\sqrt{2} & b\sqrt{2} & \sqrt{b\sqrt{5}} & \sqrt{b\sqrt{5}}\\
b & -b & -b & b\sqrt{2} & -b\sqrt{2} & a\sqrt{2} & a\sqrt{2} & 0 & 0 & 2a\sqrt{b\sqrt{5}} & -2a\sqrt{b\sqrt{5}}\\
-a & a & -a & 0 & \text{0} & b\sqrt{2} & -b\sqrt{2} & a\sqrt{2} & -a\sqrt{2} & 0 & 0
\end{array}...\\
 &  & \left.\cdots\begin{array}{cccc}
28 & 29 & 30 & 31\\
\\
0 & 0 & 2a\sqrt{b\sqrt{5}} & -2a\sqrt{b\sqrt{5}}\\
\sqrt{b\sqrt{5}} & \sqrt{b\sqrt{5}} & 0 & 0\\
2a\sqrt{b\sqrt{5}} & -2a\sqrt{b\sqrt{5}} & \sqrt{b\sqrt{5}} & \sqrt{b\sqrt{5}}
\end{array}\right]
\end{eqnarray*}  with $a=\frac{1+\sqrt{5}}{4}$ and $b=\frac{-1+\sqrt{5}}{4}.$ 

\begin{figure} [H]
\begin{center}
\includegraphics[scale=0.9]{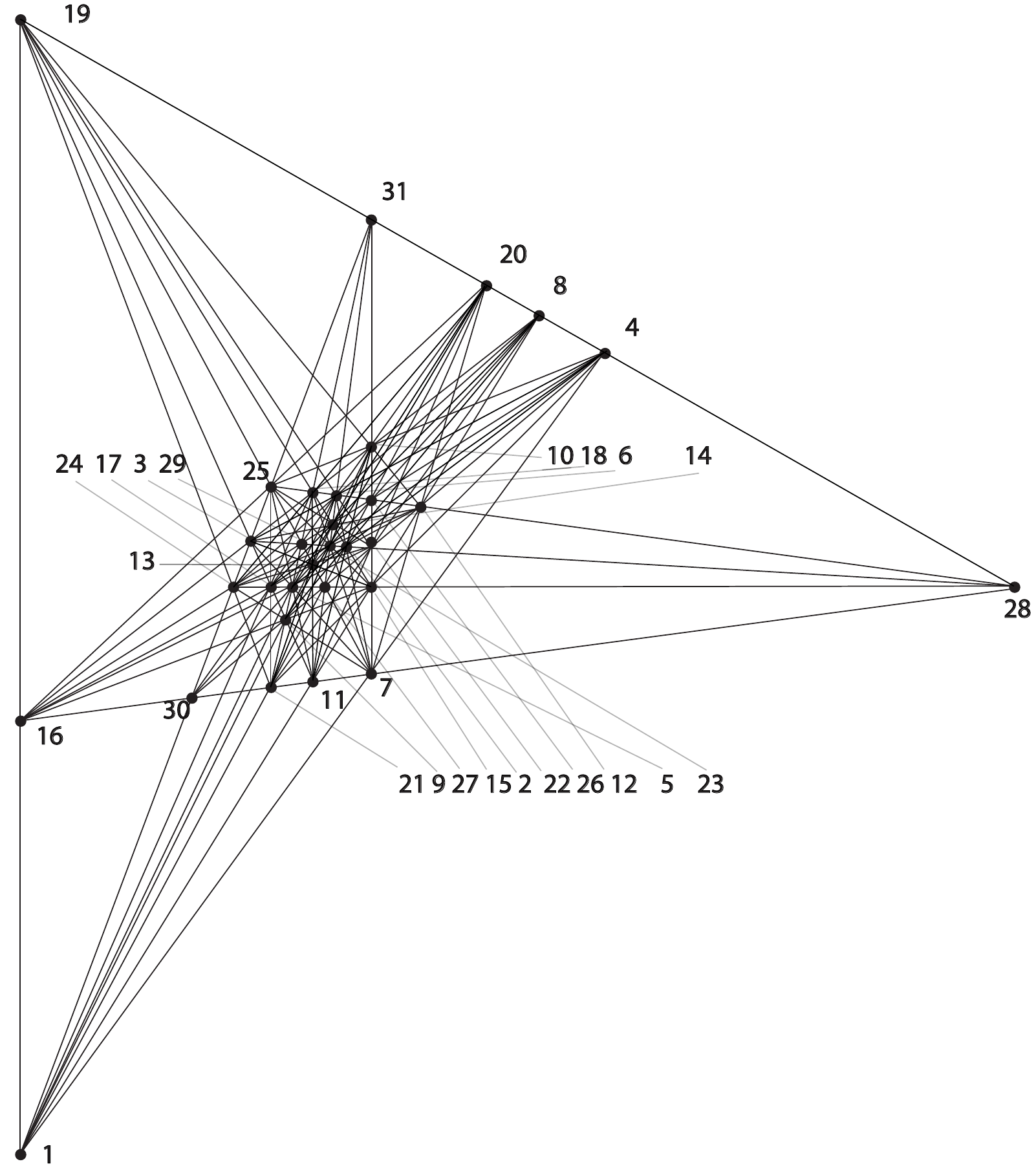}
\end{center}
\end{figure}

$$G = I$$

$\mathcal{I_{\mathrm{2}}}$ =$\{(1,22),(1,23),(1,28),(1,29),(2,24),(2,25),(2,30),(2,31),(3,20),(3,21),(3,26),(3,27),$

\qquad{}$(4,11),(4,13),(4,18),(5,10),(5,12),(5,19),(6,9),(6,15),(6,16),(7,8),(7,14),(7,17),(8,25),$

\qquad{}$(8,29),(8,27),(9,25),(9,26),(9,28),(10,24),(10,27),(10,28),(11,24),(11,26),(11,29),$

\qquad{}$(12,29),(12,31),(13,20),(13,28),(13,31),(14.21),(14,28),(12,21),(14,30),(15,20),$

\qquad{}$(15,29),(15,30),(16,23),(16,27),(16,31),(17,22),(17,26),$

\qquad{}$(17,31),(18,22),(18,27),(18,30),(19.23),(19.26),(19,30)\}$

\medskip

$\mathcal{I_{\mathrm{3}}}$ =$\begin{cases}
\{(1,4,7),(1,5,6),(2,4,5),(2,6,7),(3,4,6),(3,5,7),(4,10,25),\\
(4,15,21),(4,17,23),(5,11,25),(17,21,25),(8,21,22),(18,20,24),\\
(16,20,25),(19,21,24),(5,14,20),(5,16,22),(6,8,24),(6,13,21),\\
(6,19,22),(7,9,24),(7,12,20),(7,18,23),(9,20,23),(10,21,23),\\
(11,20,22),(12,23,24),(13,22,24),(14,22,25),(15,23,25)\} & \nu_{31}=\frac{2}{15} \\
\{(1,16,19),(1,17,18),(2,8,9),(2,10,11),(3,12,14),(3,13,15),\\
(8,14,17),(11,13,18),(9,15,16),(10,12,19)\} & \nu_{32}=\frac{1}{10} 
\end{cases}$

$\mathcal{I}_{5}$ =$\{(1,8,11,12,15),(1,9,10,13,14),(2,12,13,16,17),(2,14,15,18,19),$

\qquad{}$(3,8,10,16,18),(3,9,11,17,19)\} \hfill \nu_{5}=\frac{1}{6}$
\medskip{}

$\mathcal{I}_{6}$ =$\{(1,2,20,21,26,27),(1,3,24,25,30,31),(2,3,22,23,28,29),(4,8,19,20,28,31),$

\qquad{}$(5,8,13,23,26,30),(5,9,18,21,29,31),(5,15,17,24,27,28),(6,10,17,20,29,30),$

\qquad{}$(6,11,14,23,27,31),(6,12,18,25,26,28),(4,9,12,22,27,30),(4,14,16,24,26,29),$

\qquad{}$(7,10,15,22,26,31),(7,11,16,21,28,30),(7,13,19,25,27,29)\} \hfill \nu_{6}=\frac{1}{3}$

\medskip

\end{document}